\documentclass[10pt,journal]{IEEEtran}
\IEEEoverridecommandlockouts
\usepackage{color}

\usepackage{placeins}
\usepackage{amsthm}
\usepackage{amsmath}
\usepackage{mathtools}
\usepackage{amssymb}
\usepackage{algorithmic}
\usepackage{algorithm}
\usepackage{verbatim}
\usepackage{cite}
\usepackage{url}
\usepackage{cases}
\newtheorem{lemma}{Lemma}

\newtheorem{theorem}{Theorem}
\newtheorem{corollary}{Corollary}

\IEEEoverridecommandlockouts

\begin{document}
\title{Sub-Nyquist Cyclostationary Detection for Cognitive Radio} 
\author{Deborah Cohen, \emph{Student IEEE}
        and Yonina C. Eldar, \emph{Fellow IEEE} \thanks{
This project is funded by the European Union's Horizon 2020 research and innovation program under grant agreement No. 646804-ERC-COG-BNYQ, and by the Israel Science Foundation under Grant no. 335/14. Deborah Cohen is grateful to the Azrieli Foundation for the award of an Azrieli Fellowship. }}
\maketitle

\begin{abstract}
Cognitive Radio requires efficient and reliable spectrum sensing of wideband signals. In order to cope with the sampling rate bottleneck, new sampling methods have been proposed that sample below the Nyquist rate. However, such techniques decrease the signal to noise ratio (SNR), deteriorating the performance of subsequent energy detection. Cyclostationary detection, which exploits the periodic property of communication signal statistics, absent in stationary noise, is a natural candidate for this setting. 
In this work, we consider cyclic spectrum recovery from sub-Nyquist samples, in order to achieve both efficiency and robustness to noise. To that end, we propose a structured compressed sensing algorithm, that extends orthogonal matching pursuit to account for the structure imposed by cyclostationarity. Next, we derive a lower bound on the sampling rate required for perfect cyclic spectrum recovery in the presence of stationary noise and show that it can be reconstructed from samples obtained below Nyquist, without any sparsity constraints on the signal. In particular, in non sparse settings, the cyclic spectrum can be recovered at $4/5$ of the Nyquist rate. If the signal of interest is sparse, then the sampling rate may be further reduced to $8/5$ of the Landau rate.
Once the cyclic spectrum is recovered, we estimate the number of transmissions that compose the input signal, along with their carrier frequencies and bandwidths. Simulations show that cyclostationary detection outperforms energy detection in low SNRs in the sub-Nyquist regime. This was already known in the Nyquist regime, but is even more pronounced at sub-Nyquist sampling rates.
\end{abstract}

\IEEEpeerreviewmaketitle

\section{Introduction}

Spectrum sensing has been thoroughly investigated in the signal processing literature. Several sensing schemes have been proposed, with different performance and complexity levels. The simplest approach is energy detection \cite{Urkowitz_energy}, which does not require any a priori knowledge on the input signal. Unfortunately, energy detection is very sensitive to noise and performs poorly in low signal to noise ratio (SNR) regimes. In contrast, matched filter (MF) detection \cite{MF1, MF2}, which correlates a known waveform with the input signal to detect the presence of a transmission, is the optimal linear filter for maximizing SNR in the presence of additive stochastic noise. This technique requires perfect knowledge of the potential received transmission. A compromise between both methods is cyclostationary detection \cite{Gardner_review, cyclo_review3, GardnerBook}. This approach is more robust to noise than energy detection but at the same time only assumes the signal of interest exhibits cyclostationarity.

Cyclostationary processes have statistical characteristics that vary periodically, arising from the underlying data modulation mechanisms, such as carrier modulation, periodic keying or pulse modulation. The cyclic spectrum, a characteristic function of such processes, exhibits spectral peaks at certain frequency locations called cyclic frequencies, which are determined by the signal parameters, particularly the carrier frequency and symbol rate \cite{GardnerBook}. When determining the presence or absence of a signal, cyclostationary detectors exploit one fundamental property of the cyclic spectrum: stationary noise and interference exhibit no spectral correlation. Non-stationary interference can be distinguished from the signal of interest provided that at least one cyclic frequency of the signal is not shared with the interference \cite{GardnerBook}. This renders such detectors highly robust to noise and interference.

The traditional task of spectrum sensing has recently been facing new challenges due, to a large extent, to cognitive radio (CR) applications \cite{Mitola}. Today, CRs are perceived as a potential solution to the spectrum over-crowdedness issues, bridging between the scarcity of spectral resources and their sparse nature \cite{cr_review}. Even though most of the spectrum is already owned and new users can hardly find free frequency bands, various studies \cite{Study1, Study2, study3} have shown that it is typically significantly underutilized. CRs allow secondary users to opportunistically use the licensed spectrum when the corresponding primary user (PU) is not active \cite{Mitola}. CR requirements dictate new challenges for its most crucial task, spectrum sensing. On the one hand, detection has to be performed in real time, efficiently and with minimal software and hardware resources. On the other hand, it has to be reliable and able to cope with low SNR regimes. 

Nyquist rates of wideband signals, such as those CRs deal with, are high and can even exceed today's best analog-to-digital converters (ADCs) front-end bandwidths. In addition, such high sampling rates generate a large number of samples to process, affecting speed and power consumption. In order to efficiently sample sparse wideband signals, several new sampling methods have recently been proposed \cite{Mishali_multicoset, Mishali_theory, MagazineMishali, SamplingBook} that reduce the sampling rate in multiband settings below the Nyquist rate. This alleviates the burden on both the analog and digital sides by enabling the use of cheaper and lower power reduced rate ADCs and the processing of fewer samples.

The authors of \cite{Mishali_multicoset} derive the minimal sampling rate allowing for perfect signal reconstruction in noise-free settings and provide specific sampling and recovery techniques. However, when the final goal is spectrum sensing and detection, reconstructing the original signal is unnecessary. Power spectrum reconstruction from sub-Nyquist samples is considered in \cite{Davies, Leus, cohen2013cognitive}. 
These works seek power spectrum estimates from low rate samples, using multicoset sampling \cite{Davies, Leus, cohen2013cognitive} and the modulated wideband converter (MWC) \cite{cohen2013cognitive} proposed in \cite{Mishali_theory}. The presence or absence of a signal in a particular frequency band is then assessed with respect to the estimated power within the band. Unfortunately, the sensitivity of energy detection used in the above works is amplified when performed on sub-Nyquist samples due to noise aliasing \cite{Castro}. Therefore, this scheme fails to meet CR performance requirements in low SNR regimes. On the other hand, little a priori knowledge can be assumed on the received signals, making MF difficult to implement. Consequently, cyclostationary detection is a natural candidate for spectrum sensing from sub-Nyquist samples in low SNRs. 

Signal detection using cyclostationarity and its application to spectrum sensing for CR systems in the Nyquist regime, has been thoroughly investigated; see e.g., \cite{cyclo_review1, cr_review, cyclo_review2, cyclo_review3}. Recently, cyclostationary detection from sub-Nyquist samples was treated in \cite{ICC12, Palicot1, Palicot2, TianFeature, TianFeatureJ, cabric_cyc, TianLeus, LeusCyclo, McClellan}. A general framework is adopted, that exploits a linear relation between the sub-Nyquist and Nyquist samples, over a finite sensing time. In particular, a transformation between the Nyquist cyclic spectrum and the time-varying correlations of the sub-Nyquist samples is derived to retrieve the former from the latter. In \cite{ICC12}, the carrier frequencies, symbol periods and modulation types of the transmissions are assumed to be known. In this case, the cyclic spectrum can be reduced to its potential non zero cyclic frequencies, which are recovered using simple least squares (LS). However, this scheme cannot be applied in the context of blind spectrum sensing for CR.

The authors in \cite{Palicot1, Palicot2} consider low-pass compressive measurements of the correlation function. The cyclic autocorrelation is then reconstructed at a given lag using compressed sensing (CS) algorithms \cite{CSBook}. Two heuristic detection techniques are developed in order to infer the presence or absence of a transmission. However, it is not clear how the signal itself should be sampled in order to obtain these correlation measurements. In addition, no requirements on the number of samples for recovering the cyclic autocorrelation or guarantees on the detection methods are provided.

In \cite{TianFeature, TianFeatureJ}, the authors consider a random linear relation between the sub-Nyquist samples autocorrelation and Nyquist cyclic spectrum and formulate a $\ell_1$-norm regularized LS problem, enforcing sparsity on the latter. The same ideas are adopted in \cite{cabric_cyc} but the reconstruction is performed in matrix form, allowing for higher resolution. In \cite{TianLeus}, correlation lags beyond lag zero are exploited and the span of the random linear projections is extended beyond one period of cyclostationarity of the signal. These two extensions allow for simple LS recovery without any sparsity requirements on the signal.

The main drawback of this digital approach, adopted by all works above, is that it does not deal with the sampling scheme itself, since we do not have access to the Nyquist samples. It simply assumes that the sub-Nyquist samples can be expressed as random linear projections of the Nyquist samples. In addition, due to the inherent finite sensing time, the recovered cyclic frequencies lie on a predefined grid. Therefore, in the above works, the frequencies of interest are assumed to lie on that grid. The theoretical resolution that can be achieved is thus dictated by the sensing time. Moreover, no theoretical guarantees on the minimal sampling rate allowing for perfect recovery of the cyclic spectrum have been given. 

In \cite{LeusCyclo}, a concrete sampling scheme is considered, known as multicoset, or non-uniform sampling. The authors derive conditions on the system matrix to have full rank, allowing for perfect cyclic spectrum recovery from the compressive measurements. Random sampling in the form of a successive approximation ADC (SAR-ADC) architecture \cite{McClellanSAR} is used in \cite{McClellan}. The resulting sampling matrix is similar to multicoset, with the distinction that the grid is that of the quantization time rather than the Nyquist grid. While explicit sampling schemes are considered here, the theoretical cyclic spectrum resolution still depends on the sensing time, so that the gridding, or discretization, is part of the theoretical derivations.

In this work, we propose to reconstruct the signal's cyclic spectrum from sub-Nyquist samples obtained using the methods of \cite{Mishali_theory, Mishali_multicoset, MagazineMishali}. Our theoretical approach does not involve gridding or discretization and the cyclic spectrum can be recovered at any frequency. In addition, the MWC analog front-end presented in \cite{Mishali_theory} is a practical sampling scheme that has been implemented in hardware \cite{mwc_hardware}. We perform cyclostationarity detection on the sub-Nyquist samples, thereby obtaining both an efficient, fast and frugal detector on the one hand and one that is reliable and robust to noise on the other. We derive a sampling rate bound allowing for perfect cyclic spectrum recovery in our settings, for sparse and non sparse signals. We note that the cyclic spectrum can be perfectly recovered in the presence of stationary noise, from compressed samples, except for a limited number of cyclic frequencies that are multiples of the basic low sampling rate. For those, the reconstruction is performed in the presence of bounded noise. 

In particular, we show that, in the presence of stationary noise, the cyclic spectrum can be reconstructed from samples obtained at $4/5$ of the Nyquist rate, without any sparsity assumption on the signal. If the signal of interest is sparse, then the sampling rate can be further reduced to $8/5$ of the Landau rate, which is the Lebesgue measure of the occupied bandwidth \cite{LandauCS}. Similar results were observed in \cite{cohen2013cognitive} in the context of power spectrum reconstruction of stationary signals. There, it was shown that the power spectrum of non sparse signals can be retrieved at half the Nyquist rate and that of sparse signals can be perfectly recovered at the Landau rate. Once the cyclic spectrum is reconstructed, we apply our feature extraction algorithm, presented in \cite{liad_cyclo} in the Nyquist regime, that estimates the number of transmissions and their respective carrier frequencies and bandwidths. 

The main contributions of this paper are as follows:
\begin{itemize}
\item \textbf{Low rate sampling and digital processing} - the cyclic spectrum is recovered directly from sub-Nyquist samples. Both sampling and digital processing are performed at a low rate.
\item \textbf{Structured CS algorithm} - an orthogonal matching pursuit (OMP) based algorithm to reconstruct the cyclic spectrum is proposed that exploits the inherent structure of the correlation matrices between frequency samples of the signal.
\item \textbf{Robust detection in low SNR} - we show that cyclostationary detection performed by estimating the transmissions carrier frequency and bandwidth is more robust to noise than energy detection at sub-Nyquist rates. 
\item \textbf{Minimal sampling rate derivation} - a lower bound on the sampling rate required for cyclic spectrum recovery is derived for both sparse and non sparse signals.
\end{itemize}

This paper is organized as follows. In Section \ref{sec:model}, we describe the cyclostationary multiband model. Sections \ref{sec:samp} and \ref{sec:rec} present the sub-Nyquist sampling stage and cyclic spectrum reconstruction algorithm and conditions, respectively. Numerical experiments are presented in Section \ref{sec:sim}.

\section{Cyclostationary Multiband Model}
\label{sec:model}

\subsection{Multiband Model}
\label{sec:model1}

Let $x(t)$ be a real-valued continuous-time signal, supported on $\mathcal{F} = [-1/2T_{\text{Nyq}}, +1/2T_{\text{Nyq}}]$ and composed of up to $N_{\text{sig}}$ uncorrelated cyclostationary transmissions corrupted by additive noise, such that 
\begin{equation}
\label{eq:xmodel}
x(t)=\sum_{i=1}^{N_{\text{sig}}} s_i(t)+n(t).
\end{equation}
Here $n(t)$ is a wide-sense stationary bandpass noise and $s_i(t)$ is a zero-mean cyclostationary bandpass process, as defined below, from the class of pulse-amplitude modulation (PAM) signals:
\begin{multline} \label{eq:smodel}
s_i(t)=\sqrt{2} \cos (2 \pi f_i t) \sum_k a_{ik}^I g_i(t-kT_i) \\ - \sqrt{2} \sin (2 \pi f_i t) \sum_k a_{ik}^Q g_i(t-kT_i).
\end{multline}
The unknown symbols modulating the in-phase and quadrature components are denoted $\{a_{ik}^I\}$ and $\{a_{ik}^Q\}$, respectively, and $g_i(t)$ are the unknown pulse shape functions. The single-sided bandwidth, the carrier frequency and the symbol period are denoted by $B_i$, $f_i$ and $T_i$, respectively. Special cases of passband PAM include phase-shift keying (PSK), amplitude and phase modulation (AM-PM) and quadrature amplitude modulation (QAM) \cite{CommBook}.

Formally, the Fourier transform of $x(t)$, defined by
\begin{equation}
X(f)=\lim_{T \rightarrow \infty} \frac{1}{\sqrt{T}}\int_{-T/2}^{T/2}x(t)e^{-j2\pi f t} \mathrm{d} t,
\end{equation}
is zero for every $f \notin \mathcal{F}$. We denote by $f_{\text{Nyq}} = 1/T_{\text{Nyq}}$ the Nyquist rate of $x(t)$. The number of transmissions $N_{\text{sig}}$, their carrier frequencies, bandwidths, symbol rates and modulations, including the symbols $\{a_{ik}\}$ and the pulse shape functions $g_i(t)$ are unknown, namely the reconstruction of the cyclic spectrum, defined in the next section, is performed in a blind scenario. The single-sided bandwidth of each transmission is only assumed to not exceed a known maximal bandwidth $B$, namely $B_i \leq B$ for all $1 \leq i \leq N_{\text{sig}}$. If the bandwidth is fully occupied, then $N_{\text{sig}} B$ is on the order of $f_{\text{Nyq}}$.  

We will consider the special case of sparse multiband signals as well and show that the sampling rate for perfect cyclic spectrum reconstruction can be further reduced. In this setting, the number of transmissions $N_{\text{sig}}$ which dictates the signal sparsity, or at least an upper bound on it, is assumed to be known and $N_{\text{sig}} B \ll f_{\text{Nyq}}$. We denote by $K=2N_{\text{sig}}$ the upper bound on the number of occupied bands, to account for both the positive and negative frequency bands for each signal.


\subsection{Cyclostationarity}
\label{sec:cyclo_bg}
A process $s(t)$ is said to be cyclostationary with period $T_0$ in the wide sense if its mean $\mathbb{E}[s(t)] = \mu_s(t)$ and autocorrelation $\mathbb{E}[s(t-\tau/2)s(t+\tau/2)] = R_s(t,\tau)$ are both periodic with period $T_0$ \cite{Gardner_review}:
\begin{equation}
\mu_s(t+T_0)= \mu_s(t), \qquad R_s(t+T_0,\tau) = R_s(t,\tau).
\end{equation}
Given a wide-sense cyclostationary random process, its autocorrelation $R_s(t,\tau)$ can be expanded in a Fourier series
\begin{equation}
R_s(t,\tau) = \sum_{\alpha}R_s^{\alpha}(\tau) e^{j 2\pi \alpha t},
\end{equation}
where $\alpha=m/T_0, m \in \mathbb{Z}$ and the Fourier coefficients, referred to as cyclic autocorrelation functions, are given by
\begin{equation}
\label{cyclic_auto} R_s^{\alpha}(\tau) = \frac{1}{T_0} \int_{-T_0/2}^{T_0/2} R_s(t,\tau)  e^{-j 2\pi \alpha t}\mathrm{d}t.
\end{equation}
The cyclic spectrum is obtained by taking the Fourier transform of (\ref{cyclic_auto}) with respect to $\tau$, namely
\begin{equation}
\label{SCF} S_s^{\alpha}(f)=\int_{-\infty}^{\infty} R_s^{\alpha}(\tau) e^{-j 2\pi f \tau}\mathrm{d}\tau,
\end{equation}
where $\alpha$ is referred to as the cyclic frequency and $f$ is the angular frequency \cite{Gardner_review}. If there is more than one fundamental cyclic frequency $1/T_0$, then the process $s(t)$ is said to be polycyclostationary in the wide sense. In this case, the cyclic spectrum contains harmonics (integer multiples) of each of the fundamental cyclic frequencies \cite{GardnerBook}. These cyclic frequencies are related to the transmissions carrier frequencies and symbol rates as well as the modulation type.

An alternative interpretation of the cyclic spectrum, which we will exploit, expresses it as the cross-spectral density $S_s^{\alpha}(f)=S_{uv}(f)$ of two frequency-shifted versions of $s(t)$, $u(t)$ and $v(t)$, such that
\begin{eqnarray}
u(t) & \triangleq & s(t) e^{-j \pi \alpha t}, \\
v(t) & \triangleq & s(t) e^{+j \pi \alpha t}.
\end{eqnarray}
Then, from \cite{Papoulis}, it holds that
\begin{equation}
\label{eq:scf2}
S_s^{\alpha}(f)=S_{uv}(f)=\mathbb{E} \left[ S \left(f+\frac{\alpha}{2} \right) S^*\left(f-\frac{\alpha}{2}\right)\right].
\end{equation}

The Fourier transform of a wide-sense stationary process $n(t)$ is a white noise process, not necessarily stationary, so that \cite{Papoulis}
\begin{equation}
\label{eq:papoulis_statio}
\mathbb{E} \left[ N(f_1) N^*(f_2)\right] = S_n^{0}(f_1) \delta(f_1 -f_2).
\end{equation}
Therefore, stationary noise exhibits no cyclic correlation \cite{GardnerBook}, that is
\begin{equation}
S_n^{\alpha}(f)=0, \quad \alpha \neq 0.
\end{equation}
This property is the motivation for cyclostationary detection, in low SNR regimes in particular.

\subsection{Cyclic Spectrum of Multiband Signals}

Denote by $[f_i^{(1)}, f_i^{(2)}]$ the right-side support of the $i$th transmission $s_i(t)$. Then, $B_i=f_i^{(2)}-f_i^{(1)}$ and $f_i=(f_i^{(1)}+f_i^{(2)})/2$. The support region in the $(f,\alpha)$ plane of the cyclic spectrum $S^{\alpha}_{s_i}(f)$ of such a bandpass cyclostationary signal is composed of four diamonds, as shown in Fig.~\ref{fig:diamonds}. More precisely, it holds that \cite{GardnerBook}
\begin{equation}
\label{eq:supp}
S^{\alpha}_{s_i}(f)=0, \quad \text{for } \left||f|-\frac{|\alpha|}{2}\right| \leq f_i^{(1)} \text{or } |f| + \frac{|\alpha|}{2} \geq f_i^{(2)}.
\end{equation}
Moreover, since $x(t)$ is bandlimited to $\mathcal{F}$, it follows that \cite{GardnerBook}
\begin{equation}
\label{eq:low}
S^{\alpha}_{x}(f)=0, \quad \text{for } \quad |f| + \frac{|\alpha|}{2} \geq \frac{f_{\text{Nyq}}}{2}.
\end{equation}

\begin{figure}[tb]
  \begin{center}
    \includegraphics[width=0.5\columnwidth]{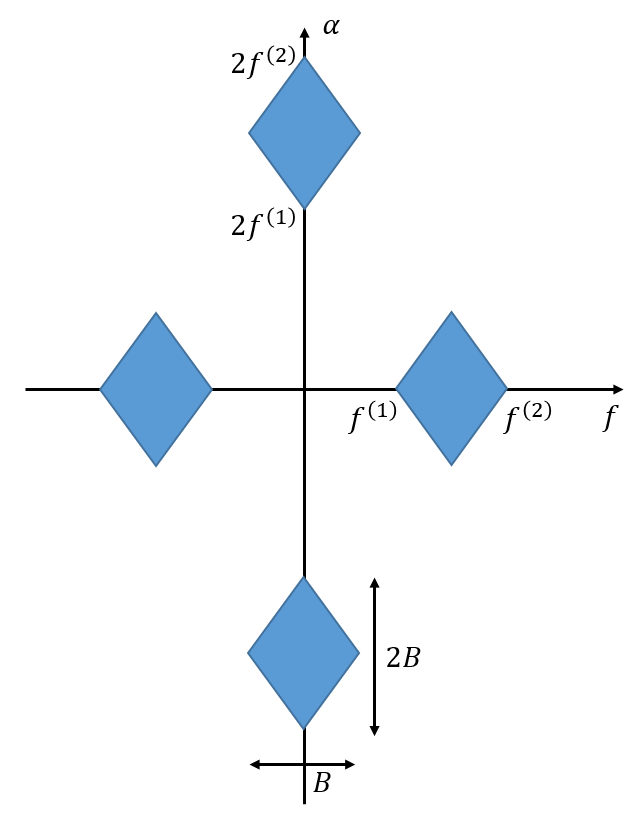}
    \caption{Support region of the cyclic spectrum of $S^{\alpha}_{s_i}(f)$.}
    \label{fig:diamonds}
  \end{center}
\end{figure}

Since the transmissions $s_i(t)$ are assumed to be zero-mean and uncorrelated (coming from different sources), the cyclic spectrum of $x(t)$ does not contain any additional harmonics which would result from correlation between different transmissions. It is thus given by
\begin{equation}
\label{eq:cyc_noise}
S_x^{\alpha}(f)= \left\{ \begin{array}{ll} \sum\limits_{i=1}^{N_{\text{sig}}} S^{\alpha}_{s_i}(f) & \alpha \neq 0 \\
\sum\limits_{i=1}^{N_{\text{sig}}} S^{0}_{s_i}(f) + S_n^0(f) & \alpha = 0. \end{array} \right.
\end{equation}
At the cyclic frequency $\alpha = 0$, the cyclic spectrum reduces to the power spectrum and the occupied bandwidth along the angular frequency axis is $2N_{\text{sig}}B$, to account for both positive and negative frequency bands. This cyclic frequency $\alpha=0$ contains the noise power spectrum as well. As a consequence, we choose to detect the transmissions and estimate their carriers and bandwidths at cyclic frequencies $\alpha \neq 0$. Note that, for $\alpha \neq 0$, the sum only contains the contributions of the transmissions that exhibit cyclostationarity at the corresponding cyclic frequency $\alpha$. Therefore, for each $\alpha \neq 0$, the sum typically contains less than $N_{\text{sig}}$ non zero elements.
It follows from (\ref{eq:cyc_noise}) that, besides the noise contribution at the cyclic frequency $\alpha = 0$, the support of $S^{\alpha}_{x}(f)$ is composed of $4N_{\text{sig}}$ diamonds, that is four diamonds for each transmission, as those shown in Fig.~\ref{fig:diamonds}.

The support of $S_x^{\alpha}(f)$ consists of two types of correlations: the two diamonds lying on the angular frequency $f$ axis and those lying on the cyclic frequency $\alpha$.
The diamonds on the $f$ axis contain self-correlations between a band and its shifted version. These correspond to cyclic peaks at locations $(f, \alpha)$, with $f \in \mathcal{F}$ and $0 \leq |\alpha| \leq B$. Cyclic features at these locations are the result of the transmissions symbol rate $1/T_i$. The value of $T_i$ can be derived from the bandwidth, as
\begin{equation}
T_i=\frac{1+\gamma_i}{B_i},
\end{equation}
where $0 \leq \gamma_i \leq 1$ is the unknown excess-bandwidth parameter of $g_i(t)$ \cite{CommBook}. Therefore, cyclic peaks that stem from the symbol rate appear at cyclic frequencies $\alpha=\frac{1}{T_i}\leq B_i \leq B$. 

The two diamonds lying on the $\alpha$ axis contain cross-correlations between two symmetric bands, belonging to the same transmission. These correspond to cyclic peaks at locations $(f, \alpha)$, with $0 \leq |f| \leq B/2$ and $B < |\alpha| \leq f_{\text{Nyq}}$. In particular, applying (\ref{eq:scf2}) for $\alpha=\pm 2f_i$, namely
\begin{equation}
S_x^{\pm 2f_i}(f)=\mathbb{E} \left[ X(f+f_i) X^*(f-f_i)\right],
\end{equation}
computes the correlation between the positive and negative bands of the $i$th transmission. This generates a peak, or cyclic feature, in the cyclic spectrum at location $(f=0, \alpha=\pm2f_i)$. By detecting the peak location, one can estimate the carrier frequency of the corresponding transmission $s_i(t)$. 
Furthermore, from (\ref{eq:supp}), it is clear that the occupied bandwidth in the cyclic spectrum at $\alpha=\pm 2f_i$ is equal to the bandwidth $B_i$. These observations are the key to estimating the transmissions carrier frequency and bandwidth.
We note that other modulations not included in (\ref{eq:smodel}) may not exhibit features at $\alpha=\pm 2f_i$. In such cases, higher order cyclostationary statistics may be used for detection \cite{GardnerBook, lim2015compressive}.

\subsection{Goal}

Our objective is to reconstruct $S_x^{\alpha}(f)$ from sub-Nyquist samples without any \emph{a priori} knowledge on the support and modulations of $s_i(t), 1 \leq i \leq N_{\text{sig}}$. We show that the cyclic spectrum of non sparse signals can be recovered from samples obtained at $4/5$ of the Nyquist rate and if the signal is assumed to be sparse, then the sampling rate can be as low at $8/5$ of the Landau rate. 

We then estimate the number of transmissions $N_{\text{sig}}$ present in $x(t)$, their carrier frequencies $f_i$ and their bandwidths $B_i$. Once the signals' carrier frequency and bandwidth are estimated, the occupied support is determined. In \cite{liad_cyclo}, we proposed a generic feature extraction algorithm for the estimation of $N_{\text{sig}}$, $f_i$ and $B_i$, for all $1 \leq i \leq N_{\text{sig}}$, from the cyclic spectrum obtained from Nyquist samples. Here, we apply the same scheme to the reconstructed cyclic spectrum from sub-Nyquist rather than Nyquist samples.

\section{Sub-Nyquist Sampling}
\label{sec:samp}

In this section, we briefly describe the sub-Nyquist sampling schemes we adopt. We consider two different sampling methods: multicoset sampling \cite{Mishali_multicoset} and the MWC \cite{Mishali_theory} which were previously proposed for sparse multiband signals in conjunction with energy detection. We show that both techniques lead to identical expressions of the signal cyclic spectrum in terms of correlations between the samples. Therefore, the cyclic spectrum reconstruction stage presented in Section \ref{sec:rec} can be applied to either of the sampling approaches.

\subsection{Multicoset sampling}
\label{sec:multico}
Multicoset sampling \cite{Bresler} can be described as the selection of certain samples from the uniform grid. More precisely, the uniform grid is divided into blocks of $N$ consecutive samples, from which only $M$ are kept. The $i$th sampling sequence is defined as
\begin{equation}
x_{c_i}[n]= \left\{ \begin{array}{ll}
x(nT_{\text{Nyq}}), & n=mN+c_i, m \in \mathbb{Z} \\
0, & \text{otherwise},
\end{array} \right.
\end{equation}
where $0 < c_1 < c_2 < \dots < c_M < N-1$. 
Multicoset sampling can be implemented as a multi-channel system with $M$ channels, each composed of a delay unit corresponding to the coset $c_i$ followed by a low rate ADC.
Let $f_s = \frac{1}{NT_{\text{Nyq}}} \ge B$ be the sampling rate of each channel and $\mathcal{F}_s=[0, f_s]$.
Following the derivations in \cite{Mishali_multicoset}, we obtain
\begin{equation}
\mathbf{z}(\tilde{f}) = \mathbf{A} \mathbf{x}(\tilde{f}), \qquad \tilde{f} \in \mathcal{F}_s,
\label{eq:multico}
\end{equation}
where $\mathbf{z}_i(\tilde{f}) = X_{c_i}(e^{j2\pi \tilde{f} T_{\text{Nyq}}}), 0 \le i \le M-1$, are the discrete-time Fourier transforms (DTFTs) of the multicoset samples and
\begin{equation}
\mathbf{x}_k(\tilde{f})=X\left(\tilde{f}+K_kf_s \right), \quad 1 \le k \le N,
\label{xdef}
\end{equation}
where $K_k = k-\frac{N+1}{2}, 1 \le k \le N$ for odd $N$ and $K_k = k-\frac{N+2}{2}, 1 \le k \le N$ for even $N$. Each entry of $\mathbf{x}(f)$ is referred to as a slice of the spectrum of $x(t)$.
The $ik$th element of the $M \times N$ matrix $\mathbf{A}$ is given by
\begin{equation}
\mathbf{A}_{ik} = \frac{1}{NT_{\text{Nyq}}} e^{j\frac{2 \pi}{N} c_i K_k},
\end{equation}
namely $\bf A$ is determined by the known sampling pattern $\{c_i\}_{i=1}^M$.

\subsection{MWC sampling}
The MWC \cite{Mishali_theory} is composed of $M$ parallel channels. In each channel, an analog mixing front-end, where $x(t)$ is multiplied by a mixing function $p_i(t)$, aliases the spectrum, such that each band appears in baseband, as illustrated in Fig.~\ref{fig:zAx}. The mixing functions $p_i(t)$ are required to be periodic with period $T_p$ such that $f_p=1/T_p \ge B$.
The function $p_i(t)$  has a Fourier expansion
\begin{equation}
p_i(t) =\sum_{l=-\infty}^{\infty} c_{il} e^{j\frac{2\pi}{T_p} lt}.
\end{equation}
In each channel, the signal goes through a lowpass filter with cut-off frequency $f_s/2$ and is sampled at the rate $f_s \ge f_p $, resulting in the samples $y_i[n]$. For the sake of simplicity, we choose $f_s=f_p$. From \cite{Mishali_theory}, the relation between the known DTFTs of the samples $y_i[n]$ and the unknown $X(f)$ is given by
\begin{equation} \label{eq:mwc}
\mathbf{z}(\tilde{f})=\mathbf{A}\mathbf{x}(\tilde{f}), \qquad \tilde{f} \in \mathcal{F}_s,
\end{equation}
where $\mathbf{z}(\tilde{f})$ is a vector of length $M$ with $i$th element $\mathbf{z}_i(\tilde{f})=Y_i(e^{j2\pi \tilde{f}T_s})$. The entries of the unknown vector $\mathbf{x}(\tilde{f})$ are given by (\ref{xdef}). The $M \times N$ matrix $\mathbf{A}$ contains the known coefficients $c_{il}$ such that
\begin{equation}
\mathbf{A}_{il} = c_{i,-l}=c^*_{il},
\end{equation}
where $N=\left\lceil f_{\text{Nyq}}/f_s \right\rceil$.
This relation is illustrated in Fig.~\ref{fig:zAx}.
\begin{figure}[h]
	\begin{center}
		\includegraphics[width=1\columnwidth]{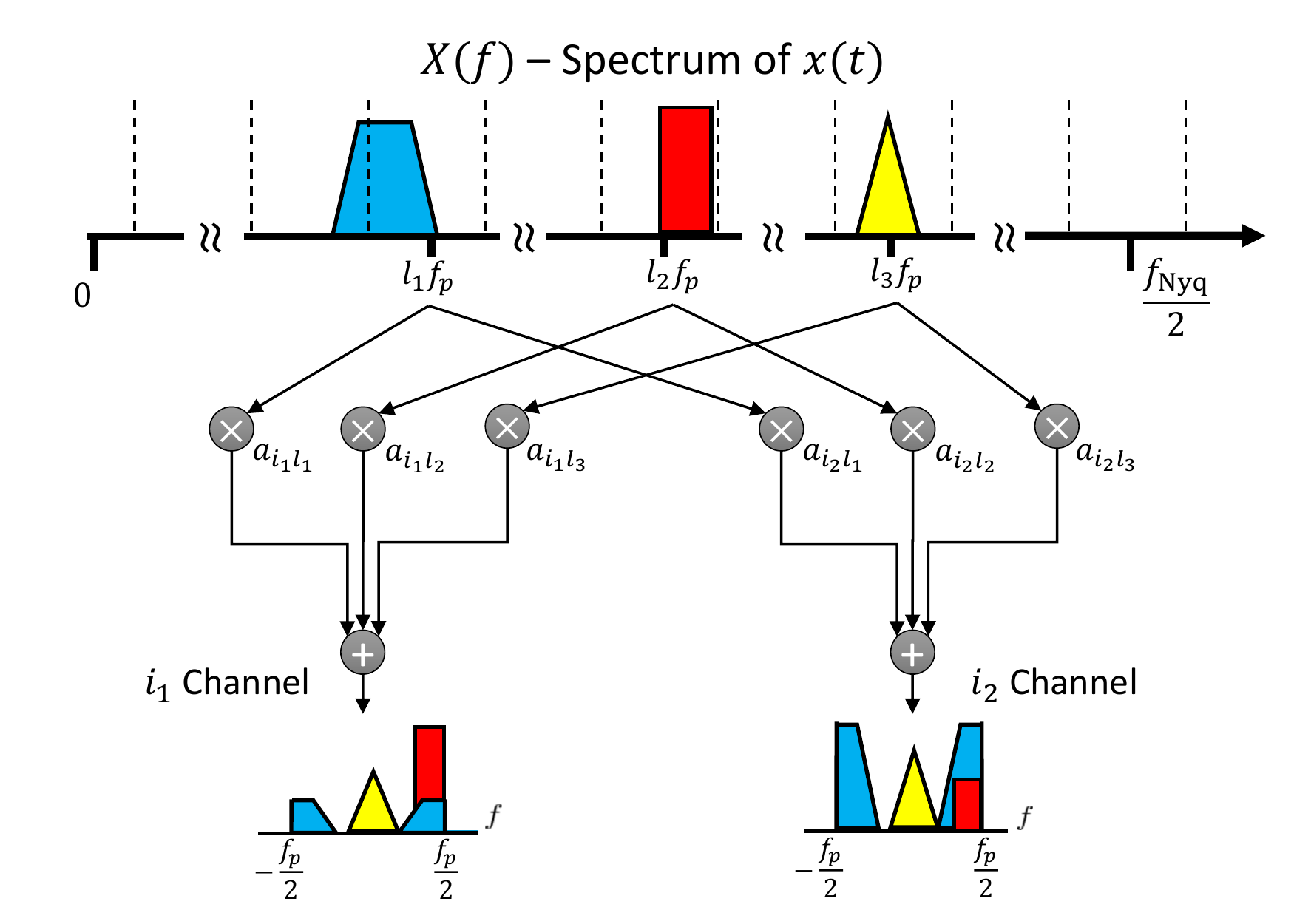}
		\caption{The spectrum slices $\mathbf{x}(\tilde{f})$ of the input signal are shown here to be multiplied by the coefficients $a_{il}$ of the sensing matrix $\mathbf{A}$, resulting in the measurements $\mathbf{z}_i(\tilde{f})$ for the $i$th channel. Note that in multicoset sampling, only the slices' complex phase is modified by the coefficients $a_{il}$. In the MWC, both the phases and amplitudes are affected in general.}
		\label{fig:zAx}
	\end{center}
\end{figure}

Systems (\ref{eq:multico}) and (\ref{eq:mwc}) are identical for both sampling schemes: the only difference is the sampling matrix $\mathbf{A}$. In the next section, we derive conditions for the reconstruction of the cyclic spectrum from either class of samples. The requirements on the resulting sampling matrix $\bf A$ are tied to conditions on the sampling pattern for multicoset sampling \cite{Mishali_multicoset} and on the mixing sequences for the MWC \cite{Mishali_theory}. We then present a method for reconstruction of the analog cyclic spectrum for both sampling schemes jointly. In particular, we will reconstruct $S_x^{\alpha}(f)$ from correlations between shifted versions of $\mathbf{z}(\tilde{f})$, defined in (\ref{eq:multico}) and (\ref{eq:mwc}). We note that for both sampling approaches, the overall sampling rate is
\begin{equation}
f_{tot}=Mf_s=\frac{M}{N}f_{\text{Nyq}}.
\end{equation}
In the simulations, we consider samples obtained using the MWC. However, multicoset samples can be used indifferently.

\section{Cyclic Spectrum Reconstruction}
\label{sec:rec}

In this section, we provide a method to reconstruct the cyclic spectrum $S_x^{\alpha}(f)$ of $x(t)$ from sub-Nyquist samples obtained using either of the sampling schemes described above, namely multicoset and the MWC. We also investigate cyclic spectrum recovery conditions, that is the minimal sampling rate allowing for perfect recovery of $S_x^{\alpha}(f)$ in the presence of stationary noise. Finally, we consider the special case of power spectrum reconstruction, presented in \cite{cohen2013cognitive}, and compare it to cyclic spectrum in terms of both recovery method and conditions.

\subsection{Relation between Samples and Cyclic Spectrum}

From (\ref{eq:multico}) or (\ref{eq:mwc}), we have
\begin{equation}
\mathbf{R}_\mathbf{z}^a(\tilde{f}) = \mathbf{A} \mathbf{R}_\mathbf{x}^a(\tilde{f}) \mathbf{A}^H, \quad \tilde{f} \in \left[0, f_s-a \right],
\label{eq:autoco2}
\end{equation}
for all $a \in [0, f_s]$, where
\begin{equation}
\mathbf{R}_\mathbf{x}^a(\tilde{f}) = \mathbb{E} \left[ \mathbf{x}(\tilde{f}) \mathbf{x}^H(\tilde{f}+a) \right],
\end{equation}
and
\begin{equation} \label{eq:rza}
\mathbf{R}_\mathbf{z}^a(\tilde{f})  = \mathbb{E} \left[\mathbf{z}(\tilde{f}) \mathbf{z}^H(\tilde{f}+a) \right].
\end{equation}
Here $(.)^H$ denotes the Hermitian operation. The entries in the matrix $\mathbf{R}_\mathbf{x}^a(\tilde{f})$ are correlations between shifted versions of the slices $\mathbf{x}(\tilde{f})$, namely correlations between frequency-shifted versions of $x(t)$. The variable $a$ controls the shift between the slices, while $\tilde{f}$, obtaining values in the interval $[0, f_s-a]$, determines the specific frequency location within the slice. Both $a$ and $\tilde{f}$ are continuous and $\mathbf{R}_\mathbf{x}^a(\tilde{f})$ can be computed for any combination of $a \in [0,f_s]$ and $\tilde{f} \in [0, f_s-a]$. In practice, with a limited sensing time, $a$ and $\tilde{f}$ are discretized according to the number of samples per channel.

To proceed, we begin by investigating the link between the cyclic spectrum $S_x^{\alpha}(f)$ and the shifted correlations between the slices $\mathbf{x}(\tilde{f})$, namely the entries of $\mathbf{R}_{\mathbf{x}}^a(\tilde{f})$. We then show how the latter can be recovered from $\mathbf{R}_{\mathbf{z}}^a(\tilde{f})$ using (\ref{eq:autoco2}).


The alternative definition of the cyclic spectrum (\ref{eq:scf2}) implies that the elements in the matrix $\mathbf {R}_\mathbf{x}^a(\tilde{f})$ are equal to $S_x^{\alpha}(f)$ at the corresponding $\alpha$ and $f$. Indeed, it can easily be shown that
\begin{equation}
\label{eq:mapping}
\mathbf{R}_\mathbf{x}^a(\tilde{f})_{(i,j)} = S_x^{\alpha}(f),   \\ 
\end{equation}
for
\begin{eqnarray}
\label{eq:mapping_param}
\alpha&=&(j-i)f_s+a \nonumber \\ f&=&-\frac{f_{\text{Nyq}}}{2} + \tilde{f}- \frac{f_s}{2}+ \frac{(j+i)f_s}{2}+\frac{a}{2}. 
\end{eqnarray}
Here $\mathbf{R}_\mathbf{x}^a(\tilde{f})_{(i,j)}$ denotes the $(i,j)$th element of $\mathbf{R}_\mathbf{x}^a(\tilde{f})$.
In particular, from (\ref{eq:cyc_noise}), it follows that for $a=0$, namely with no shift, it holds that 
\begin{equation} \label{eq:cyc0}
\mathbf{R}_\mathbf{x}^0(\tilde{f})=\sum\limits_{k=1}^{N_{\text{sig}}} \mathbf{R}_{\mathbf{s}_k}^0(\tilde{f}) + \mathbf{R}_\mathbf{n}^0(\tilde{f}).
\end{equation}
Here, $\mathbf{R}_\mathbf{n}^0(\tilde{f})$ is a diagonal matrix that contains the noise's power spectrum, such that
\begin{equation}
\mathbf{R}_\mathbf{n}^0(\tilde{f})_{kk}=S_n^0 \left( -\frac{f_{\text{Nyq}}}{2} +kf_s+\tilde{f} \right),
\end{equation}
for $0 \leq k \leq N-1$, and
\begin{equation}
\mathbf{R}_{\mathbf{s}_k}^a(\tilde{f}) =\mathbb{E} \left[ \mathbf{s}_k(\tilde{f}) \mathbf{s}_k^H(\tilde{f}+a) \right],
\end{equation}
where $\mathbf{s}_k(\tilde{f})$ is a vector of size $N$ whose non zero elements are the frequency slices from $\mathbf{x}(\tilde{f})$ corresponding to the $k$th transmission.
The diagonal of $\mathbf{R}_\mathbf{x}^0(\tilde{f})_{kk}$, for $0 \leq k \leq N-1$, contains the power spectrum of $x(t)$ such that
\begin{equation}
\mathbf{R}_\mathbf{x}^0(\tilde{f})_{kk}=S_x^0 \left( -\frac{f_{\text{Nyq}}}{2} +kf_s+\tilde{f} \right),
\end{equation}
is the sum of the transmissions' and noise's power spectrum. Since $n(t)$ is stationary, for $a \neq 0$, we have
\begin{equation}
\mathbf{R}_\mathbf{x}^a(\tilde{f})=\sum\limits_{k=1}^{N_{\text{sig}}} \mathbf{R}_{\mathbf{s}_k}^a(\tilde{f}).
\end{equation}
Our goal can then be stated as recovery of $\mathbf{R}_\mathbf{x}^a(\tilde{f})$, for $a \in [0,f_s]$ and $\tilde{f} \in [0, f_s-a]$, since once $\mathbf{R}_\mathbf{x}^a(\tilde{f})$ is known, $S_x^{\alpha}(f)$ follows for all $(\alpha,f)$, using (\ref{eq:mapping}).

We now consider the structure of the autocorrelation matrices $\mathbf{R}_\mathbf{x}^a(\tilde{f})$, which is related to the support of the cyclic spectrum $S_x^{\alpha}(f)$. In Section \ref{sec:cyclo_bg}, we discussed the support of $S_x^{\alpha}(f)$, composed of two types of correlations: self-correlations between a band and its shifted version and cross-correlations between shifted versions of symmetric bands belonging to the same transmission. Consider a given frequency location $\tilde{f}$ and shift $a$. The frequency component $\mathbf{x}_i(\tilde{f})$, for $0 \leq i \leq N-1$, can be correlated to at most two entries of $\mathbf{x}(\tilde{f}+a)$, one from the same band and one from the symmetric band. The correlated component can be either in the same, respectively symmetric, slice or in one of the adjacent slices. This follows from the fact that each band may split between two slices at most, since we require $f_p \geq B$. Thus, the first correlated entry is either $i$ or $i \pm 1$ and the second is either $N-i$, $N-i+1$ or $N-1+2$. Since the noise is assumed to be wide-sense stationary, from (\ref{eq:papoulis_statio}), a noise frequency component is correlated only with itself. Thus, $n(t)$ can contribute non-zero elements only on the diagonal of $\mathbf{R}_\mathbf{x}^0(\tilde{f})$.

Figures~\ref{fig:schema_0} and \ref{fig:schema_a} illustrate these correlations for $a=0$ and $a=f_s/2$, respectively. First, in Fig.~\ref{fig:schema_x}, an illustration of the spectrum of $x(t)$, namely $X(f)$, is presented for the case of a sparse signal buried in stationary bandpass noise. It can be seen that frequency bands of $X(f)$ either appear in one $f_p$-slice or split between two slices. The resulting vector of spectrum slices $\mathbf{x}(f)$ and the correlations between these slices without any shift, namely $\mathbf{R}_\mathbf{x}^0(\tilde{f})$, are shown in Figs.~\ref{fig:schema_0}(a) and (b), respectively. Define the $d$-diagonal and $d$-anti diagonal of an $N \times N$ matrix to be its $(i,j)$th elements such that $j=i+d$, and $j=N-i+1-d$, respectively. In particular, the $0$-diagonal stands for the main diagonal and the $0$-anti diagonal is the secondary diagonal. In Fig.~\ref{fig:schema_0}(b), we observe that self-correlations appear only on the main diagonal since every frequency component is correlated with itself. In particular, the main diagonal contains the noise's power spectrum (in green). Cross-correlations between the yellow symmetric triangles appear in the $0$-anti diagonal, whereas those of the blue trapeziums are contained in the $-1$ and $1$-anti diagonals. The red rectangles do not contribute any cross-correlations for a shift of $a=0$. Figures~\ref{fig:schema_a}(a) and (b) show the vector $\mathbf{x}(\tilde{f})$ and its shifted version $\mathbf{x}(\tilde{f}+a)$ for $a=f_s/2$, respectively. The resulting correlation matrix $\mathbf{R}_\mathbf{x}^a(\tilde{f})$ appears in Fig.~\ref{fig:schema_a}(c). Here, the self correlations of the blue trapezium appear in the $-1$-diagonal. The non zero cross-correlations all appear in the anti-diagonal, for the shift $a=f_s/2$. We note that for this shift, the yellow triangles do not contribute self or cross correlations.

\begin{figure}
  \begin{center}
    \includegraphics[width=0.9\columnwidth]{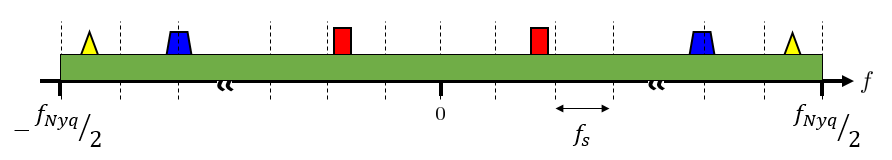}
    \caption{Original spectrum $X(f)$.}
    \label{fig:schema_x}
  \end{center}
\end{figure}

\begin{figure}
  \begin{center}
    \includegraphics[width=0.85\columnwidth]{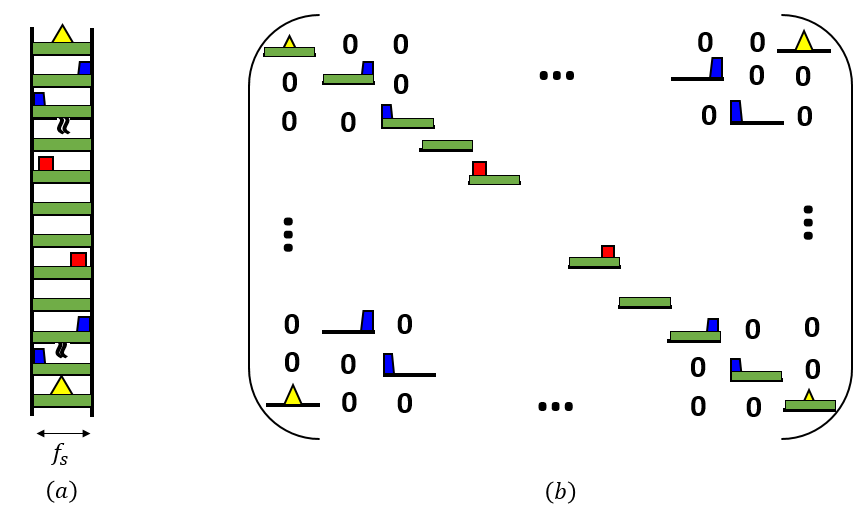}
    \caption{(a) Spectrum slices vector $\mathbf{x}(\tilde{f})$ (b) correlated slices of $\mathbf{x}(\tilde{f})$ in the matrix $\mathbf{R}_\mathbf{x}^0(\tilde{f})$.}
    \label{fig:schema_0}
  \end{center}
\end{figure}

\begin{figure}
  \begin{center}
    \includegraphics[width=0.85\columnwidth]{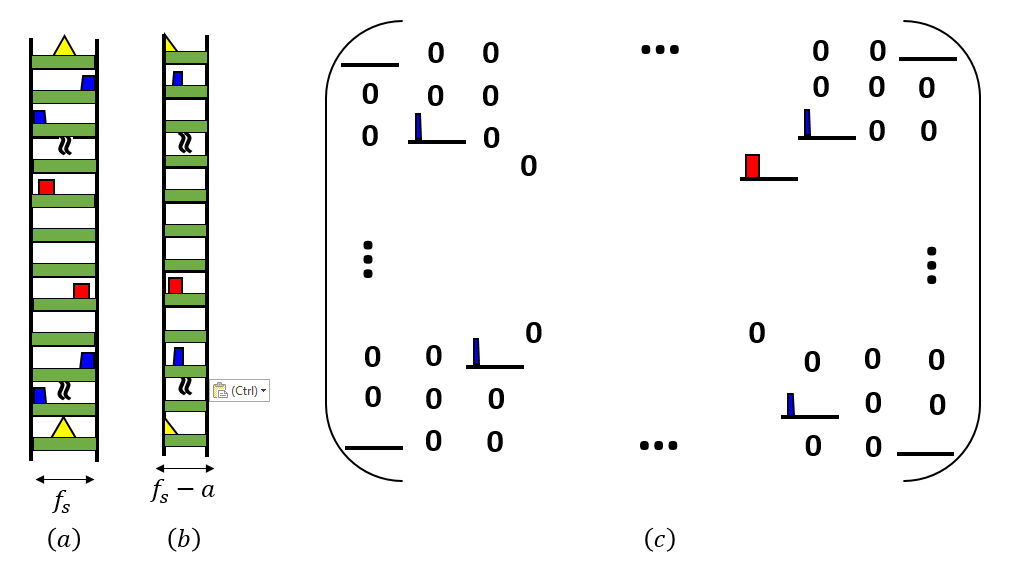}
    \caption{(a) Spectrum slices vector $\mathbf{x}(\tilde{f})$ (b) spectrum slices shifted vector $\mathbf{x}(\tilde{f}+a)$ for $a=f_s/2$ (c) correlated slices of $\mathbf{x}(\tilde{f})$ and $\mathbf{x}(\tilde{f}+a)$ in the matrix $\mathbf{R}_\mathbf{x}^a(\tilde{f})$, with $a =f_s/2$.}
    \label{fig:schema_a}
  \end{center}
\end{figure}

The following four conclusions can be drawn from the observations above on the structure of $\mathbf{R}_\mathbf{x}^a(\tilde{f})$, for a given $a \in ]0,f_s]$ and $\tilde{f} \in [0,f_s-a]$. We will treat the case where $a=0$ separately since it yields a different structure due to the presence of noise.
\begin{itemize}
\item \textbf{Conclusion 1:} The non zeros entries of $\mathbf{R}_{\mathbf{x}}^a(\tilde{f})$ are contained in its $-1$, $0$ and $1$-diagonals and $-1$, $0$ and $1$-anti diagonals.



\item \textbf{Conclusion 2:} The $i$th row of $\mathbf{R}^a_{\mathbf{x}}(\tilde{f})$ contains at most two non zero elements at locations $\left(i,g(i)\right)$ and $\left(i,g(N-i+1)\right)$, where
\begin{equation} \label{eq:g_def}
g(i)= \left\{ \begin{array}{ll} i \text{ or } i \pm 1 & 2 \leq i \leq N-1, \\
i \text{ or } i+1 & i=1, \\
i \text{ or } i-1  & i=N.
\end{array} \right.
\end{equation}
\item \textbf{Conclusion 3:} The $j$th column of $\mathbf{R}_{\mathbf{x}}^a(\tilde{f})$ contains at most two non zero elements at locations $\left(g(j),j) \right)$ and $\left(g(N-j+1), j \right)$, for $1 \leq j \leq N$.
\item \textbf{Conclusion 4:} For each specific frequency $\tilde{f}$, a transmission contributes to at most two slices, one in the negative and one in the positive frequencies. This is due to the assumption that $f_p \geq B$. Therefore, $\mathbf{R}_\mathbf{x}^a(\tilde{f})$ contains at most $K=2N_{\text{sig}}$ rows/columns that have non zero elements. Without any sparsity assumption, it is obvious that $K=N$ even if the number of transmissions is greater than $N/2$.
\end{itemize}
From \textbf{Conclusions 2} and \textbf{4}, it follows that $\mathbf{R}_\mathbf{x}^a(\tilde{f})$ is $2K$-sparse and has additional structure described in \textbf{Conclusions 1-3}.

Since the non zero elements of $\mathbf{R}_\mathbf{x}^a(\tilde{f})$ only lie on the 3 main and anti-diagonals, (\ref{eq:autoco2}) can be further reduced to
\begin{equation}
\mathbf{r}_{\mathbf{z}}^a(\tilde{f}) = \mathbf{(\bar{A} \otimes A)}
\text{vec} (\mathbf{R}_{\mathbf{x}}^a(\tilde{f}))  =  \mathbf{(\bar{A} \otimes A)} \mathbf{B}\mathbf{r}_{\mathbf{x}}^a(\tilde{f})  \triangleq \mathbf{\Phi} \mathbf{r}_{\mathbf{x}}^a(\tilde{f}) ,
\label{eq:rzrx}
\end{equation}
where $\bf \bar{A}$ denotes the conjugate matrix of $\bf A$ and
\begin{equation} \label{eq:Phi_def}
\bf \Phi=(\bar{A} \otimes A)B.
\end{equation}
Here $\otimes$ is the Kronecker product, $\mathbf{r}_{\mathbf{z}}^a(\tilde{f})  = \text{vec}(\mathbf{R}_{\mathbf{z}}^a(\tilde{f}))$, where $\text{vec}(\cdot)$ denotes the column stack concatenation operation, and $\bf B$ is a $N^2 \times (6N-4)$ selection matrix that selects the elements of the $-1$, $0$ and $1$-diagonals and anti-diagonals of $\mathbf{R}_\mathbf{x}^a(\tilde{f})$ from the vector $\text{vec} (\mathbf{R}_{\mathbf{x}}^a(\tilde{f}))$. The resulting $(6N-4) \times 1$ vector composed of these selected elements is denoted by $\mathbf{r}_{\mathbf{x}}^a(\tilde{f})$.

From \textbf{Conclusions 2-4}, the vector $\mathbf{r}_{\mathbf{x}}^a(\tilde{f})$ is $2K$-sparse and its support presents additional structure. Denote by $\Sigma_k$ the set of $k$-sparse vectors that belong to a linear subspace $\Sigma$.
In our case, $\mathbf{r}_{\mathbf{x}}^a(\tilde{f})$ belongs to $\Sigma_k=\Sigma_{2K}$, defined as
\begin{equation}
\label{eq:Sig_def}
\Sigma_{2K} = \{ \mathbf{x} \in \mathbb{R}^{(6N-4) \times 1} | \left|S(\mathbf{x})\right| \leq 2K \wedge S(\mathbf{x})  \in \mathcal{I} \},
\end{equation}
where $S(\mathbf{x})$ denotes the support of $\mathbf{x}$ and the set $\mathcal{I}$ is determined by the following properties:
\begin{enumerate}
\item If the group indexed by $j$ is in $\mathcal{I}$, then the group indexed by $(N+1-j)$ is in $\mathcal{I}$ as well. The groups are defined in the following items.
\item If the group indexed by $2 \leq j \leq N-1$ is in $\mathcal{I}$, then it means that at most one of the entries $\{6j-1, 6j, 6j+1\}$ and at most one of the entries $\{6j+2, 6j+3, 6j+4\}$ are in $\mathcal{I}$. 
\item If the group indexed by $j=1$ is in $\mathcal{I}$, then at most one of the entries $\{1, 2\}$ and at most one of the entries $\{3, 4\}$ are in $\mathcal{I}$. 
\item If the group indexed by $j=N$ is in $\mathcal{I}$, then at most one of the entries $\{6N-7, 6N-6\}$ and at most one of the entries $\{6N-5, 6N-4\}$ are in $\mathcal{I}$. 
\end{enumerate}
Item 1) follows from \textbf{Conclusion 4} and items 2)-3) and 4) follow from \textbf{Conclusions 2-3}.
We note that the indexation above is valid if the selection matrix in (\ref{eq:rzrx}) selects the elements of $\text{vec} (\mathbf{R}_{\mathbf{x}}^a(\tilde{f}))$ row by row. Any other selection order would yield a different indexation.


Consider now the case in which $a=0$. The matrix $\mathbf{R}_{\mathbf{x}}^0(\tilde{f})$ contains the power spectrum of the transmissions and the noise on its main diagonal and cyclic components of the transmissions on its $-1$, $0$ and $1$-anti diagonals. The remaining elements are zero.
Denote by $\mathbf{r}_{\mathbf{x}_1}^0$ the $N \times 1$ non-sparse vector composed of the diagonal of $\mathbf{R}_{\mathbf{x}}^0(\tilde{f})$ and by $\mathbf{r}_{\mathbf{x}_2}^0$ the $(3N-2) \times 1$ sparse vector composed of its $-1$, $0$ and $1$-anti diagonals. Let $\mathbf{\Phi}_1= \mathbf{\bar{A}} \odot \mathbf{A}$ where $\odot$ denotes the Khatri-Rao product and $\mathbf{\Phi}_2= (\mathbf{\bar{A}} \otimes \mathbf{A})\mathbf{\tilde{B}}$ with $\bf \tilde{B}$ a $N^2 \times (3N-2)$ selection matrix that selects the elements of the $-1$, $0$ and $1$-anti diagonals of $\mathbf{R}_\mathbf{x}^0(\tilde{f})$ from the vector $\text{vec} (\mathbf{R}_{\mathbf{x}}^0(\tilde{f}))$.
Combining (\ref{eq:cyc0}) and (\ref{eq:rzrx}), we can write $\mathbf{r}_\mathbf{z}^0(\tilde{f})$ as the sum of two components as follows,
\begin{equation} \label{eq:cyc_vec0}
\mathbf{r}_\mathbf{z}^0(\tilde{f})= \mathbf{\Phi}_2 \mathbf{r}_{\mathbf{x}_2}^0(\tilde{f}) + \mathbf{w}(\tilde{f}),
\end{equation}
where $\mathbf{w}(\tilde{f})=\mathbf{\Phi}_1 \mathbf{r}_{\mathbf{x}_1}^0(\tilde{f})$ is the noise component. We note that the signal's power spectrum is buried in the noise $\mathbf{w}(\tilde{f})$. Therefore, we do not recover it and signal detection will be performed only on cyclic frequencies $\alpha \neq 0$.

From (\ref{eq:low}) and (\ref{eq:mapping}), by recovering $\mathbf{r}_\mathbf{x}^a(\tilde{f}) \in \Sigma_{2K}$ for all $a \in [0, f_s], \tilde{f} \in [0, f_s-a]$, we recover the entire cyclic spectrum of $x(t)$. We consider only $a \geq 0$ and consequently $\alpha \geq 0$. We thus only reconstruct half of the cyclic spectrum, known to be symmetric \cite{GardnerBook}. In (\ref{eq:rzrx}), there is no noise component, even if the signal $x(t)$ is corrupted by additional stationary noise. For the corresponding cyclic frequencies, we can therefore achieve perfect recovery. In contrast, in (\ref{eq:cyc_vec0}), namely for $a=0$, there is an additional noise component. From (\ref{eq:mapping_param}), this case corresponds to cyclic frequencies which are multiples of the channels' sampling frequency $f_s$. For these frequencies, the recovery of the sparse vector $\mathbf{r}_{\mathbf{x}_2}^0(\tilde{f}) \in \Sigma_{2K}$ is not perfect and is performed in the presence of bounded noise. In the simulations, we observe that for detection purposes, this noisy recovery is satisfactory. To achieve perfect recovery for these cyclic frequencies as well, we may sample the signal using a different sampling frequency $f_{s_2}$.

\subsection{Cyclic Spectrum Recovery Conditions}

We now consider conditions for perfect recovery of the cyclic spectrum $S_x^{\alpha}(f)$ from sub-Nyquist samples. Corollary \ref{thm:general} below derives sufficient conditions on the minimal number of channels $M$ for perfect recovery of $\mathbf{R}_\mathbf{x}^a(\tilde{f})$, for any $a \in ]0,f_s]$ and $\tilde{f} \in [0, f_s-a]$ in the presence of additive stationary noise. As stated above, for $a=0$, the recovery is noisy. From \textbf{Conclusion 1}, we only need to recover $\mathbf{r}_{\mathbf{x}}^a(\tilde{f})$ from (\ref{eq:rzrx}) or (\ref{eq:cyc_vec0}) for $a \neq 0$ or $a=0$, respectively. Theorem \ref{prop:general} first states sufficient conditions for reconstruction of the vector $\mathbf{r}_\mathbf{x}^a(\tilde{f}) \in \Sigma_{2K}$. To that end, we rely on the following Lemma which is well known in the CS literature \cite{SamplingBook, CSBook}.


\begin{lemma}
\label{th:struct_sparse}
For any vector $\mathbf{y} \in \mathbb{R}^m$, there exists at most one signal $\mathbf{x} \in \Sigma_k$ such that $\bf y=Ax$ if and only if all sets of $2k$ columns of $\bf A$ belonging to $W$, such that
\begin{equation} \label{eq:W_def}
W = \{ S(\mathbf{x}_1) \bigcup S(\mathbf{x}_2) | \mathbf{x}_1, \mathbf{x}_2 \in \Sigma_k \},
\end{equation}
are linearly independent. In particular, for uniqueness we must have that $m \geq 2k$.
\end{lemma}

%

Using Lemma \ref{th:struct_sparse}, it follows that in order to perfectly recover $\mathbf{r}_x^a(\tilde{f}) \in \Sigma_{2K}$ from $\mathbf{r}_z^a(\tilde{f})$, we need to ensure that all sets of $4K$ columns of $\bf \Phi$ belonging to $W$ are linearly independent. In our case, $W$ is the set of unions of the supports of two vectors from $\Sigma_{2K}$, as defined in (\ref{eq:W_def}).
The following theorem relates spark properties of the sampling matrix $\bf A$ to rank properties of $M^2 \times 4K$ sub-matrices of $\bf \Phi$, whose columns belong to $W$, which we denote by $\mathbf{\Phi}^W$. Since $\mathbf{\Phi}_2$ in (\ref{eq:cyc_vec0}) is a submatrix of $\bf \Phi$, we only need to consider recovery conditions for the case $a \neq 0$. For $a=0$, under these conditions, we do not achieve perfect recovery due to the presence of noise in (\ref{eq:cyc_vec0}).

\begin{theorem} \label{prop:general}
Let $\bf \Phi$ be defined in (\ref{eq:Phi_def}), with $\bf A$ of size $M \times N$ ($M \le N$) such that $\text{spark}(\mathbf{A})=M+1$. The $M^2 \times 4K$ matrix $\mathbf{\Phi}^W$, whose columns belong to $W$ defined in (\ref{eq:W_def}) with $\Sigma_k=\Sigma_{2K}$ from (\ref{eq:Sig_def}), is full column rank if $M>\frac{8}{5}K$.
\end{theorem}
We note that if $K \geq 2$, namely there is at least one transmission, then $M > \frac{8}{5}K$ implies $M^2 \geq 4K$ and $\mathbf{\Phi}^W$ is a tall matrix.

\begin{proof}
Recall that the matrix $\bf \Phi$ is expressed as
\begin{multline}
\mathbf{\Phi} = [\begin{matrix}
\mathbf{\bar{a}}_{1} \otimes \mathbf{a}_{1} &\mathbf{\bar{a}}_{2} \otimes \mathbf{a}_{2} & \dots & \mathbf{\bar{a}}_{N} \otimes \mathbf{a}_{N} & \dots \end{matrix}  \\
 \begin{matrix} \dots & \mathbf{\bar{a}}_{1} \otimes \mathbf{a}_{2} &  \mathbf{\bar{a}}_{2} \otimes \mathbf{a}_{3} &  \dots &  \mathbf{\bar{a}}_{N-1} \otimes \mathbf{a}_{N} & \end{matrix}\\
 \begin{matrix} \dots & \mathbf{\bar{a}}_{2} \otimes \mathbf{a}_{1} &  \mathbf{\bar{a}}_{3} \otimes \mathbf{a}_{2} &  \dots &  \mathbf{\bar{a}}_{N} \otimes \mathbf{a}_{N-1} & \end{matrix}\\
 \begin{matrix}  \mathbf{\bar{a}}_{1} \otimes \mathbf{a}_{N} &\mathbf{\bar{a}}_{2} \otimes \mathbf{a}_{N-1} & \dots & \mathbf{\bar{a}}_{N} \otimes \mathbf{a}_{1} & \dots \end{matrix}  \\
 \begin{matrix} \dots & \mathbf{\bar{a}}_{1} \otimes \mathbf{a}_{N-1} &  \mathbf{\bar{a}}_{2} \otimes \mathbf{a}_{N-2} &  \dots &  \mathbf{\bar{a}}_{N-1} \otimes \mathbf{a}_{1} & \end{matrix}\\
 \begin{matrix} \dots & \mathbf{\bar{a}}_{2} \otimes \mathbf{a}_{N} &  \mathbf{\bar{a}}_{3} \otimes \mathbf{a}_{N-1} &  \dots &  \mathbf{\bar{a}}_{N} \otimes \mathbf{a}_{2} & \end{matrix}],
\end{multline}
where $\mathbf{a}_i$ and $\mathbf{\bar{a}}_i$ denote the $i$th column of $\bf A$ and its conjugate, respectively. Assume by contradiction that the columns of $\mathbf{\Phi}^W$ are linearly dependent. Then, there exist $\beta_1, \cdots \beta_{4K}$, not all zeros, such that
\begin{equation}
\label{eq:proof_syst}
\sum_{j=1}^{4K} \beta_j \boldsymbol{\phi}^W_{j}=\mathbf{0},
\end{equation}
where $\boldsymbol{\phi}^W_j$ denotes the $j$th column of $\mathbf{\Phi}^W$ and is of the form
\begin{equation}
\boldsymbol{\phi}^W_{j} = \mathbf{\bar{a}}_{[j]} \otimes \mathbf{a}_{g([j])} \text{ or } \mathbf{\bar{a}}_{[j]} \otimes \mathbf{a}_{g([N+1-j])}.
\end{equation}
Here $\mathbf{a}_{[j]}$ denotes the column of $\bf A$ that corresponds to the $j$th selected column of $\bf \Phi$ and $g(\cdot)$ is defined in (\ref{eq:g_def}). 
Denote by $k_0$ the number of indices $[j]$, for $1 \leq j \leq K$, that appear twice in $\mathbf{\Phi}^W$. Obviously, $0 \leq k_0 \leq K$ and $k_0$ is even.

From 1)-4) in the definition of $\Sigma_{2K}$, we can express $\mathbf{\Phi}^W$ as
\begin{multline}
\mathbf{\Phi}^W = [\begin{matrix}
\mathbf{\bar{a}}_{[1]} \otimes \mathbf{a}_{g([1])} & \dots & \mathbf{\bar{a}}_{[K]} \otimes \mathbf{a}_{g([K])} & \dots \end{matrix} \\
\begin{matrix}\dots & \mathbf{\bar{a}}_{[1]} \otimes \mathbf{a}_{h([1])} & \dots & \mathbf{\bar{a}}_{[k_0/2]} \otimes \mathbf{a}_{h([k_0/2])}  & \dots \end{matrix} \\
\begin{matrix}\dots &  \mathbf{\bar{a}}_{[K-k_0/2]} \otimes \mathbf{a}_{h([K-k_0/2])} & \dots & \mathbf{\bar{a}}_{[K]} \otimes \mathbf{a}_{h([K])} & \dots \end{matrix} \\
\begin{matrix}\dots & \mathbf{\bar{a}}_{[K+1]} \otimes \mathbf{a}_{g([K+1])} & \dots & \mathbf{\bar{a}}_{[2K-k_0]} \otimes \mathbf{a}_{g([2K-k_0])}  &\dots \end{matrix} \\
 \begin{matrix} \dots & \mathbf{\bar{a}}_{[1]} \otimes \mathbf{a}_{g([K])} & \dots & \mathbf{\bar{a}}_{[K]} \otimes \mathbf{a}_{g([1])} & \dots \end{matrix} \\
\begin{matrix} \dots & \mathbf{\bar{a}}_{[1]} \otimes \mathbf{a}_{h([K])} & \dots & \mathbf{\bar{a}}_{[k_0/2]} \otimes \mathbf{a}_{h([K-k_0/2])} & \dots \end{matrix} \\
\begin{matrix} \dots & \mathbf{\bar{a}}_{[k_0/2]} \otimes \mathbf{a}_{h([K-k_0/2])} & \dots & \mathbf{\bar{a}}_{[K]} \otimes \mathbf{a}_{h([1])} & \dots \end{matrix} \\
 \begin{matrix} \dots & \mathbf{\bar{a}}_{[K+1]} \otimes \mathbf{a}_{[2K-k_0]} &  \dots & \mathbf{\bar{a}}_{g([2K-k_0])} \otimes \mathbf{a}_{g([K+1])} \end{matrix}].
\end{multline}
Here, $h(i)$ is defined as $g(i)$ and we can have either $h(i)=g(i)$ or $h(i) \neq g(i)$.
By rearranging the columns of $\mathbf{\Phi}^W$ with respect to the left entry index of each Kronecker product, the system of equations (\ref{eq:proof_syst}) can be written as
\begin{equation}
\label{eq:proof_eq}
\mathbf{C} \bar{\mathbf{A}}^T= \mathbf{0},
\end{equation}
where the $M \times (2K-k_0)$ matrix $\bf C$ is defined by
\begin{multline}
\mathbf{C} = 
[\begin{matrix}  \beta_1 \mathbf{a}_{g([1])} +\beta_{K+1} \mathbf{a}_{h([1])} + \beta_{2K+1} \mathbf{a}_{g([K])} +\beta_{3K+1} \mathbf{a}_{h([K])}  \end{matrix} \\
\begin{matrix}   \cdots &  \beta_K \mathbf{a}_{g([K])} +\beta_{K+k_0} \mathbf{a}_{h([K])} + \beta_{3K} \mathbf{a}_{g([1])} +\beta_{3K+k_0} \mathbf{a}_{h([1])}  \end{matrix} \\
 \begin{matrix} \cdots & \beta_{K+k_0+1} \mathbf{a}_{g([K+1])} + \beta_{3K+k_0+1} \mathbf{a}_{g([2K-k_0])}  \end{matrix} \\
 \begin{matrix} \cdots & \beta_{2K} \mathbf{a}_{g([2K-k_0])} + \beta_{4K} \mathbf{a}_{g([K+1])}  \end{matrix}].
\end{multline}

We next show that $\text{rank}(\mathbf{C} \bar{\mathbf{A}}^T)>0$ in order to contradict the assumption that there exist $\beta_1, \cdots \beta_{4K}$, not all zeros, such that (\ref{eq:proof_syst}) holds.
Denote by $n_0$ the number of quadruplets $(\beta_j, \beta_{K+j}, \beta_{2K+j}, \beta_{3K+j})$ , for $1 \leq j \leq K$, or pairs $(\beta_{k_0+j}, \beta_{2K+k_0+j})$, for $K+1 \leq j \leq 2K-k_0$, with one non zero element at least. Obviously, $1 \leq n_0 \leq 2K-k_0$. Let $\bf \tilde{C}$ be the $M \times n_0$ matrix composed of the $n_0$ columns of $\bf C$ corresponding to these non zero quadruplets/pairs and let $\mathbf{\tilde{A}}^T$ be constructed out of the corresponding $n_0$ rows of $\bar{\mathbf{A}}^T$. Then, from the Sylvester rank inequality,
\begin{equation}
\text{rank}(\mathbf{C} \bar{\mathbf{A}}^T)=\text{rank}(\mathbf{\tilde{C}} \tilde{\mathbf{A}}^T)\geq \text{rank}(\mathbf{\tilde{C}}) + \text{rank}(\mathbf{\tilde{A}}^T) - n_0.
\end{equation}
In addition, 
\begin{equation}
\text{rank}(\mathbf{\tilde{A}}^T) \geq \text{min}\left(n_0, \text{spark}(\mathbf{A})-1 \right),
\end{equation}
and 
\begin{equation}
\text{rank}(\mathbf{\tilde{C}}) \geq \text{min}\left(\left\lceil \frac{n_0}{4} \right\rceil , \frac{\text{spark}(\mathbf{A})-1}{4} \right).
\end{equation}
The last inequality follows from the fact that any linear combination of $n$ columns of $\bf C$ is a linear combination of at least $\left\lceil n/4 \right\rceil$ distinct columns of $\bf A$. Therefore, 
\begin{equation}
\text{rank}(\mathbf{C} \bar{\mathbf{A}}^T) \geq \text{min}\left(\left\lceil \frac{n_0}{4} \right\rceil, \frac{5}{4}M - n_0 \right).
\end{equation}
Since $M>\frac{8}{5}K$, it holds that $\text{rank}(\mathbf{C} \bar{\mathbf{A}}^T) \geq 1$, for all $1 \leq n_0 \leq 2K-k_0$, $0 \leq k_0 \leq K$, contradicting (\ref{eq:proof_eq}).
\end{proof}

Corollary \ref{thm:general}, which directly follows from Theorem \ref{th:struct_sparse}, provides sufficient conditions for perfect recovery of $\mathbf{r}_x^a(\tilde{f})$.
\begin{corollary}  \label{thm:general}
\label{th:first}
If
\begin{enumerate}
\item $\text{spark}(\mathbf{A})=M+1$,
\item $M > \frac{8}{5}K$,
\end{enumerate}
then the system (\ref{eq:rzrx}) has a unique solution in $\Sigma_{2K}$.
\end{corollary}
\begin{proof}
Conditions 1)-2) ensure that all sets of $2K$ columns belonging to $W$ are linearly independent, from Theorem \ref{prop:general}. Thus, the proof follows from Lemma \ref{th:struct_sparse}.
\end{proof}
Under the conditions of Corollary \ref{th:first}, the cyclic spectrum $S_x(f)$ can be perfectly recovered for cyclic frequencies $\alpha$ which are not multiples of $f_s$. For $\alpha$ that are multiples of $f_s$, the recovery is performed in the presence of bounded noise, yielding a bounded reconstruction error. For detection purposes, this has proven satisfactory in the simulations.

Without any sparsity assumption, we can repeat the proof of Theorem \ref{prop:general} with $K=N$ and $k_0=N$, leading to $1 \leq n_0 \leq N$. We thus obtain that, if $M>\frac{4}{5}N$, then we can perfectly recover the cyclic spectrum of $x(t)$. The minimal sampling rate is then
\begin{equation}
f_{\text{min}_0}=Mf_s =\frac{4}{5}Nf_s=\frac{4}{5}f_{\text{Nyq}}.
\end{equation}
This means that even without any sparsity constraints on the signal, we can retrieve its cyclic spectrum from samples below the Nyquist rate, by exploiting its cyclostationary properties. 
A similar result was already observed in \cite{cohen2013cognitive} in the context of power spectrum reconstruction of wide-sense stationary signals in noiseless settings. In that case, the power spectrum slices appear only on the diagonal of the matrix $\mathbf{R}^0_x(\tilde{f})$ and it follows that $\bf \Phi= \bar{A} \odot A$. Power spectrum recovery is therefore a special case of cyclic spectrum reconstruction, treated here. There, it was shown that the power spectrum can be retrieved at half the Nyquist rate without any sparsity constraints. Here, we extend this result to cyclic spectrum reconstruction, which requires a higher rate.

If $x(t)$ is assumed to be sparse in the frequency domain, with $K=2N_{\text{sig}} \ll N$, then the minimal sampling rate for perfect reconstruction of its cyclic spectrum is
\begin{equation}
f_{\text{min}}=Mf_s =\frac{16}{5}N_{\text{sig}}B=\frac{8}{5}f_{\text{Landau}}.
\end{equation}
It was shown in \cite{cohen2013cognitive}, that the power spectrum of a stationary sparse signal can be perfectly recovered at its Landau rate. Again, the minimal sampling rate for cyclic spectrum recovery is slightly higher than that required for power spectrum reconstruction.

\subsection{Cyclic Spectrum Recovery}

So far, we only discussed conditions for perfect recovery of the cyclic spectrum, namely for (\ref{eq:rzrx}) and (\ref{eq:cyc_vec0}) to have unique solutions. We now provide an algorithm for cyclic spectrum reconstruction. To account for the structure of $\mathbf{r}_{\mathbf{x}}^a(\tilde{f})$ for $a \neq 0$, we extend orthogonal matching pursuit (OMP) \cite{CSBook, SamplingBook}. In each iteration, we add an internal loop that, for a selected element originally from the diagonals of $\mathbf{R}_{\mathbf{x}}^a(\tilde{f}$), checks for a corresponding non-zero element from the anti-diagonals, and vice versa, as defined by the set $\mathcal{I}$. For $a=0$, we use the standard OMP \cite{CSBook, SamplingBook}. We note that for all $a \geq 0$ we can exploit the additional symmetric structure of $\mathbf{r}_{\mathbf{x}}^a(\tilde{f})$ as defined by Property 1 of $\mathcal{I}$. Our structured OMP method (assuming that the columns of $\mathbf{\Phi}$ are normalized) is formally defined by Algorithm~\ref{alg:OMP}.

In Algorithm~\ref{alg:OMP}, $\mathbf{v} = \mathbf{r}_{\mathbf{z}}^a(\tilde{f})$, $\Lambda_i$ denotes the set of complementary indices with respect to $i$ according to Property 2 of $\mathcal{I}$, namely for $2 \leq i \leq N-1$,
\begin{equation}
\Lambda_i = \left\{ \begin{array}{ll} \{6d_i +5, 6d_i +6, 6d_i +7 \}& \text{if } 0<\text{mod} (i-4, 6) \leq 3 \\
 \{6d_i +8, 6d_i +9, 6d_i +10 \} & \text{else}, \end{array} \right.
\end{equation}
where $d_i = \lfloor \frac{i-4}{6} \rfloor$. For $i=1$ and $i=N$, $\Lambda_i$ is similarly defined according to Properties 3 and 4, respectively.
The vector $\mathbf{w}^S$ is the reduction of $\mathbf{w}$ to the support $S$, $\mathbf{\Phi}_S$ contains the corresponding columns of $\bf \Phi$, $S^c$ denotes the complementary set of $S$ and $\dagger$ is the Moore-Penrose pseudo-inverse. The halting criterion in Algorithm~\ref{alg:OMP}, as for standard OMP, can be sparsity-based if the true sparsity is known, or at least an upper bound for it, or residual-based.

\begin{algorithm}[!t]
\caption{Structured OMP \label{alg:OMP}}
\begin{algorithmic}
\STATE \textbf{Input:} observation vector $\mathbf{v}$ of size $m$, measurement matrix $\mathbf{\Phi}$ of size $m \times n$, threshold $\epsilon >0$
\STATE \textbf{Output:} index set $S$ containing the locations of the non zero indices of $\mathbf{u}$, estimate for sparse vector $\bf \hat{u}$
\STATE \textbf{Initialize:} residual $\mathbf{r}_0=\mathbf{v}$, index set $S_0=\emptyset$, possible index set $\Gamma_0 = \{1, \dots, n\}$, estimate $\mathbf{\hat{u}}=\mathbf{0}$, $\ell=0$
\WHILE{halting criterion false}
\STATE $\ell \leftarrow \ell+1$
\STATE $\mathbf{b} \leftarrow \mathbf{\Phi}^* \mathbf{r}_{\ell}$
\STATE $S_{\ell} \leftarrow S_{\ell} \cup \arg \max\limits_{i \in \Gamma_{\ell}} \mathbf{b}_i$
\STATE $\Gamma_{\ell} \leftarrow \Gamma_{\ell} \setminus \Lambda_i$
\STATE $(\mathbf{\hat{u}}_\ell)^{S_{\ell}} \leftarrow \mathbf{\Phi}_{S_{\ell}}^{\dagger} \mathbf{v}$, $(\mathbf{\hat{u}}_\ell)^{S_{\ell}^c} \leftarrow \mathbf{0}$
\STATE $\delta_0 \leftarrow ||\mathbf{v} - \mathbf{\Phi} \mathbf{\hat{u}}_{\ell}||^2$
\FOR{$j \in \Lambda_i$}
\STATE $\mathbf{w}_j^{S_{\ell}} \leftarrow \mathbf{\Phi}_{S_{\ell} \cup j}^{\dagger} \mathbf{v}$, $\mathbf{w}_j^{S_{\ell}^c} \leftarrow \mathbf{0}$
\STATE $\delta_j \leftarrow ||\mathbf{v} - \mathbf{\Phi} \mathbf{w}_{j}||^2$
\ENDFOR
\IF{$ \delta_0 - \min\limits_{j \in \Lambda_i} \delta_j > \epsilon$}
\STATE $S_{\ell} \leftarrow S_{\ell} \cup \arg \min\limits_{j \in \Lambda_i} \delta_j$
\STATE $(\mathbf{\hat{u}}_\ell)^{S_{\ell}} \leftarrow \mathbf{\Phi}_{S_{\ell}}^{\dagger} \mathbf{v}$, $(\mathbf{\hat{u}}_\ell)^{S_{\ell}^c} \leftarrow \mathbf{0}$
\ENDIF
\STATE $\mathbf{r} \leftarrow \mathbf{v} - \mathbf{\Phi} \mathbf{\hat{u}}_\ell$
\ENDWHILE
\STATE return $S_{\ell}$ and $\mathbf{\hat{u}}_{\ell}$
\end{algorithmic}
\end{algorithm}

Similarly to \cite{Mishali_multicoset}, the set (\ref{eq:rzrx}) consists of an infinite number of linear systems since $\tilde{f}$ is a continuous variable. Since the support $S$ is common to $\mathbf{r}_{\mathbf{x}}^a(\tilde{f})$ for all $\tilde{f} \in \mathcal{F}_s$, we propose to recover it jointly instead of solving (\ref{eq:rzrx}) for each $\tilde{f}$ individually, thus increasing efficiency and robustness to noise. To that end, we use the support recovery paradigm from \cite{Mishali_multicoset} that produces a finite system of equations, called multiple measurement vectors (MMV) from an infinite number of linear systems.
This reduction is performed by what is referred to as the continuous to finite (CTF) block. The cyclic spectrum reconstruction of both sparse and non sparse signals can then be divided into two stages: support recovery, performed by the CTF, and cyclic spectrum recovery. 
From (\ref{eq:rzrx}), for $a \in ]0, f_s]$, we have
\begin{equation}
\mathbf{Q}^a = \mathbf{\Phi Z}^a \mathbf{\Phi}^H
\end{equation}
where
\begin{equation}
\mathbf{Q}^a= \int_{\tilde{f} \in \mathcal{F}_s} \mathbf{r}_{\mathbf{z}}^a(\tilde{f}) {\mathbf{r}_{\mathbf{z}}^a}^H(\tilde{f}) \mathrm{d}\tilde{f}
\end{equation}
is an $M \times M$ matrix and
\begin{equation}
\mathbf{Z}^a= \int_{\tilde{f} \in \mathcal{F}_s} \mathbf{r}_{\mathbf{x}}^a(\tilde{f}) {\mathbf{r}_{\mathbf{x}}^a}^H(\tilde{f}) \mathrm{d}\tilde{f}
\end{equation}
is an $N \times N$ matrix. Then, any matrix $\mathbf{V}^a$ for which $\mathbf{Q}^a=\mathbf{V}^a \left( \mathbf{V}^a\right)^H$ is a frame for $\mathbf{r}_{\mathbf{z}}^a(\mathcal{F}_s)=\{\mathbf{r}_{\mathbf{z}}^a(\tilde{f})|\tilde{f} \in \mathcal{F}_s \}$ \cite{Mishali_multicoset, mishali2008reduce}. Clearly, there are many possible ways to select $\mathbf{V}^a$. We construct it by performing an eigendecomposition of $\mathbf{Q}^a$ and choosing $\mathbf{V}^a$ as the matrix of eigenvectors corresponding to the non zero eigenvalues. We then define the following linear system
\begin{equation} \label{eq:CTF}
\mathbf{V}^a= \mathbf{\Phi} \mathbf{U}^a.
\end{equation}
For $a=0$, identical derivations can be carried out by replacing $\bf \Phi$ by $\mathbf{\Phi}_2$.
From \cite{Mishali_multicoset} (Propositions 2-3), the support of the unique sparsest solution of (\ref{eq:CTF}) is the same as the support of $\mathbf{r}_{\mathbf{x}}^a(\tilde{f})$ in our original set of equations (\ref{eq:rzrx}) (or (\ref{eq:cyc_vec0})). For simplicity, Algorithm~\ref{alg:OMP} presents the single measurement vector (SMV) version of the recovery algorithm, which can be adapted to the MMV settings, similarly to the simultaneous OMP \cite{CSBook, SamplingBook}, to solve (\ref{eq:CTF}).

As discussed above, $\mathbf{r}_{\mathbf{x}}^a(\tilde{f})$ is $2K$-sparse for each specific $\tilde{f} \in [0, f_s-a]$, for all $a \in [0, f_s]$. However, after combining the frequencies, the matrix $\mathbf{U}^a$ is $4K$-sparse (at most), since the spectrum of each transmission may split between two slices. Therefore, the above algorithm, referred to as SBR4 in \cite{Mishali_multicoset} (for signal reconstruction as opposed to cyclic spectrum reconstruction), requires a minimal sampling rate of $2f_{\text{min}}$ for sparse signals or $2f_{\text{min}_0}$ for non sparse signals. In order to achieve the minimal rate $f_{\text{min}}$ or $f_{\text{min}_0}$, the SBR2 algorithm regains the factor of two in the sampling rate at the expense of increased complexity \cite{Mishali_multicoset}. In a nutshell, SBR2 is a recursive algorithm that alternates between the CTF described above and a bi-section process. The bi-section splits the original frequency interval into two equal width intervals on which the CTF is applied, until the level of sparsity of $\mathbf{U}^a$ is less or equal to $2K$. We refer the reader to \cite{Mishali_multicoset} for more details.

Once the support $S$ is known, perfect reconstruction of the cyclic spectrum is obtained by
\begin{eqnarray} \label{eq:recs}
(\mathbf{\hat{r}}_{\mathbf{x}}^a)^S(\tilde{f}) &=& \mathbf{\Phi}_S^{\dagger} \mathbf{r}_{\mathbf{z}}^a(\tilde{f})\\
\mathbf{\hat{r}}_{\mathbf{x}_i}^a(\tilde{f}) &=& 0 \quad \forall i \notin S, \nonumber
\end{eqnarray}
for all $a \in ]0, f_s]$ and
\begin{eqnarray} \label{eq:recs0}
(\mathbf{\hat{r}}_{\mathbf{x}}^0)^S(\tilde{f}) &=& (\mathbf{\Phi}_2)_S^{\dagger} \mathbf{r}_{\mathbf{z}}^0(\tilde{f})\\
\mathbf{\hat{r}}_{\mathbf{x}_i}^0(\tilde{f}) &=& 0 \quad \forall i \notin S. \nonumber
\end{eqnarray} 
The cyclic spectrum $S_x^{\alpha}(f)$ is then assembled using (\ref{eq:mapping}) for $(f, \alpha)$ defined in (\ref{eq:mapping_param}).

\subsection{Carrier Frequency and Bandwidth Estimation}

Once the cyclic spectrum $S_x^{\alpha}(f)$ is reconstructed from the sub-Nyquist samples, we apply our carrier frequency and bandwidth estimation algorithm from \cite{liad_cyclo}. Our approach is a simple parameter extraction method from the cyclic spectrum of multiband signals. It allows the estimation of several carriers and several bandwidths simultaneously, as well as that of the number of transmissions $N_{\text{sig}}$. The proposed algorithm consists of the following five steps: preprocessing, thresholding, clustering, parameter estimation, corrections.
Here, we briefly describe the algorithm steps. The reader is referred to \cite{liad_cyclo} for more details.

The preprocessing aims to compensate for the presence of stationary noise in the cyclic spectrum at the cyclic frequency $\alpha=0$, by attenuating the cyclic spectrum energy around this frequency. Thresholding is then applied to the resulting cyclic spectrum in order to find its peaks. For each cyclic frequency $\alpha$, we retain the value of the cyclic spectrum at $f=0$. The locations and values of the selected peaks are then clustered to find the corresponding cyclic feature. Before separating the clusters, we start by estimating their number using the elbow method, which can be traced to speculation by Thorndike \cite{elbow}. Clustering is then performed using the k-means method. At the end of the process, each cluster represents a cyclic feature. It follows that, apart from the cluster present at DC which we remove, the number of signals $N_\text{sig}$ is equal to half the number of clusters. Next, we estimate the carrier frequency $f_i$ and bandwidth $B_i$ of each transmission. The carrier frequency yields the highest correlation \cite{GardnerBook} and thus the highest peak, at the cyclic frequency equal to twice its value, namely $\alpha=2f_i$. It is therefore estimated as half the cyclic frequency of the highest peak within the clusters belonging to the same signal. The bandwidth is found by locating the edge of the support of the angular frequencies.

The processing flow of our low rate sampling and cyclic spectrum recovery algorithm is summarized in Fig.~\ref{fig:diagram}. 
\begin{figure}
  \begin{center}
    \includegraphics[width=1\columnwidth]{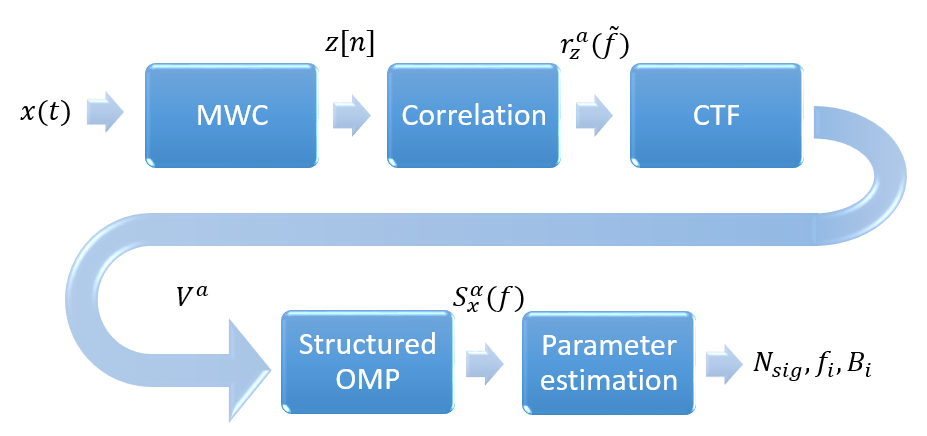}
    \caption{Processing flow diagram. The input signal $x(t)$ is first fed to the MWC analog front-end which generates the low rate samples $z[n]$ shown in (\ref{eq:mwc}) in the frequency domain. The correlations $\mathbf{r_z}^a(\tilde{f})$ between frequency shifted versions of the samples are then computed by vectorizing (\ref{eq:rza}). The CTF next produces a finite set of equations, which are solved using Algorithm~\ref{alg:OMP}, exploiting the known structure of the sparse vectors that compose the cyclic spectrum $S_x^{\alpha}(f)$. The number of transmissions $N_{\text{sig}}$, their respective carrier frequencies $f_i$ and bandwidths $B_i$ are finally estimated from $S_x^{\alpha}(f)$. The entire processing is performed at a low rate.}
    \label{fig:diagram}
  \end{center}
\end{figure}

\section{Simulation Results}
\label{sec:sim}

We now demonstrate via simulations cyclic spectrum reconstruction from sub-Nyquist samples and investigate the performance of our carrier frequency and bandwidth estimation algorithm. We compare our approach to energy detection and investigate the impact of the sampling rate on the detection performance. Throughout the simulations we use the MWC analog front-end \cite{mwc_hardware} for the sampling stage.

\subsection{Preliminaries}

We begin by explaining how we estimate the elements of $\mathbf{R}_{\mathbf{z}}^a(\tilde{f})$ in (\ref{eq:rza}). The overall sensing time is divided into $P$ time windows of length $Q$ samples. 
We first compute the estimates of $\mathbf{z}_i(\tilde{f}), 1 \le i \le M$ using the fast Fourier transform (FFT) on the samples $\mathbf{z}_i[n]$ over a finite time window. We then estimate the elements of $\mathbf{R}_{\mathbf{z}}^a(\tilde{f})$ as
\begin{equation}
\mathbf{\hat{R}_z}^a(\tilde{f})_{(i,j)}=\frac{1}{P} \sum_{p=1}^{P} \hat{\mathbf{z}}_i^p(\tilde{f})  \hat{\mathbf{z}}_j^p(\tilde{f}+a) ,
\end{equation}
for $a \in [0, f_s]$ and $f \in [0, f_s -a]$. Here, $\hat{\mathbf{z}}_i^p(\tilde{f})$ is the estimate of $\mathbf{z}_i(\tilde{f})$ from the $p$th time window. 

The cyclic spectrum recovery processing is presented here in the frequency domain. The reconstruction can be equivalently performed in the time domain by modulating the slices to replace the frequency shift $\tilde{f}+a$. Then, $\mathbf{\hat{r}_x}^a[n]$ can be recovered using the time equivalent of (\ref{eq:recs})-(\ref{eq:recs0}). However, the carrier frequencies $f_i$ and bandwidths $B_i$ estimation is performed on the cyclic spectrum, in the frequency domain. Thus, the Fourier transform of $\mathbf{\hat{r}_x}^a[n]$ needs to be computed, and $S_x^{\alpha}(f)$ is then mapped using (\ref{eq:mapping}) for $(f, \alpha)$ defined in (\ref{eq:mapping_param}). Therefore, we choose to perform the entire processing in the frequency domain. Another reason to do so is that SBR2 can obviously be performed in the frequency domain only, as opposed to SBR4 which can be carried out both in time and frequency.

We note that in theory, our approach does not require any discretization, neither in the angular frequency $f$ nor in the cyclic frequency $\alpha$. Indeed, $\mathbf{R}_x^{a}(\tilde{f})$ can be computed for any $a \in [0, f_s]$ and $\tilde{f} \in [0, f_s-a]$. This distinguishes our scheme from those based on a transformation between Nyquist and sub-Nyquist samples, where the resolution is theoretically inherent to the problem dimension and dictated by the length of the Nyquist samples vector. In practice, the resolution both in $f$ and $\alpha$ obviously depends on the sensing time and is determined by the number of samples, namely the number of discrete Fourier transform (DFT) coefficients of $\hat{\mathbf{z}}(\tilde{f})$.

We compare our cyclostationary detection to energy detection based on power spectrum recovery, presented in \cite{cohen2013cognitive}. There, it was shown that power spectrum sensing outperforms spectrum sensing, namely energy detection performed on the recovered signal itself. This power spectrum reconstruction approach is a special case of ours for $a=0$, when only the matrix $\mathbf{R}_x^0(\tilde{f})$ is considered and only its diagonal is reconstructed. Therefore, we compare our detection approach performed on $S_x^{\alpha}(f)$ for $\alpha \neq 0$ to energy detection carried out on $S_x^{\alpha}(f)$, for $\alpha=0$, corresponding to the diagonal elements of $\mathbf{R}_x^0(\tilde{f})$.

Throughout the simulations we consider additive white Gaussian noise (AWGN) $n(t)$. The SNR is defined as the ratio between the power of the wideband signal and that of the wideband noise as follows
\begin{equation}
\text{SNR}=\frac{\sum_{i=1}^{N_{\text{sig}}}||s_i(t)||^2}{||n(t)||^2}.
\end{equation}

\subsection{Cyclic Spectrum Recovery}

We first illustrate cyclic spectrum reconstruction from sub-Nyquist samples. We consider $x(t)$ composed of $N_{\text{sig}}=3$ AM transmissions. Each transmission has bandwidth $B=80$MHz and the carrier frequencies are drawn uniformly at random in $[0, \frac{f_{\text{Nyq}}}{2}]$, with $f_{\text{Nyq}}=6.4$GHz. The SNR is set to $-5$dB. In the sampling stage, we use the MWC with $M=11$ channels, each sampling at $f_s=95$MHz. The overall sampling rate is therefore $1.05$GHz, that is $2.2$ times the Landau rate and $16\%$ of the Nyquist rate. Here, the theoretical minimal sampling rate is $f_{\text{min}}=768$MHz.
Figure~\ref{fig:spec1} shows the original and reconstructed power spectrum of $x(t)$, namely $S_x^0(f)$. In this experiment, the carriers are $f_1=97$MHz, $f_2=573$MHz and $f_3=1.4$GHz. We observe that the recovery of the power spectrum failed; some occupied bands were not reconstructed and others that only contained noise were identified as active. This is due to the poor performance of energy detection in low SNR regimes. In can be seen in Fig.~\ref{fig:spec1} that the signal spectrum is not buried in noise and the transmissions would have been perfectly detected using energy detection on Nyquist samples. However, in sub-Nyquist regimes, the aliasing decreases the SNR \cite{Castro}, as can be seen in Fig.~\ref{fig:samples}, which shows the DFT of the samples of one channel. As a consequence, energy detection failed in this sub-Nyquist regime.
\begin{figure}[tb]
  \begin{center}
    \includegraphics[width=0.9\columnwidth]{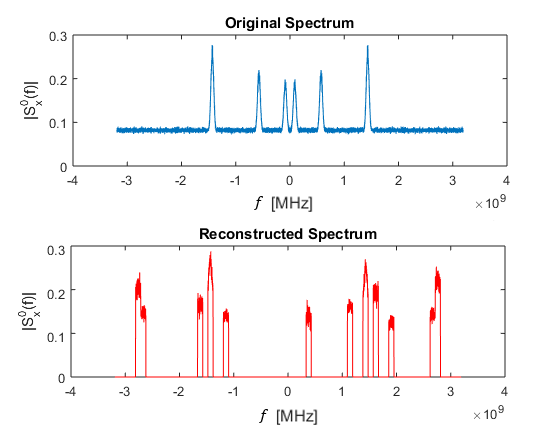}
    \caption{Original (top) and reconstructed (bottom) power spectrum ($\alpha = 0$).}
    \label{fig:spec1}
  \end{center}
\end{figure}

\begin{figure}[tb]
  \begin{center}
    \includegraphics[width=0.9\columnwidth]{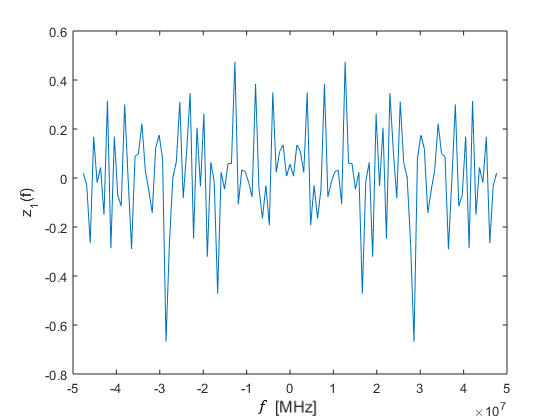}
    \caption{Sub-Nyquist samples in the first channel in the frequency domain.}
    \label{fig:samples}
  \end{center}
\end{figure}
The reconstructed cyclic spectrum of $x(t)$ is presented in Fig.~\ref{fig:cyc2d1} (the reader is referred to the colored version for a clearer figure), where we observe that the noise contribution is concentrated at $\alpha=0$ while it is significantly lower at the non zero cyclic frequencies. For clarity, we focus on the one-dimensional section of $S_x^{\alpha}(f)$ for $f=0$, shown in Fig.~\ref{fig:cyc1d1}. It can be clearly seen that the highest peaks (at least $6$dB above the lower peaks) are located at $\alpha_i=2f_i$ for all 3 active transmissions. This illustrates the advantage of cyclostationary detection in comparison with energy detection.
\begin{figure}[tb]
  \begin{center}
    \includegraphics[width=0.9\columnwidth]{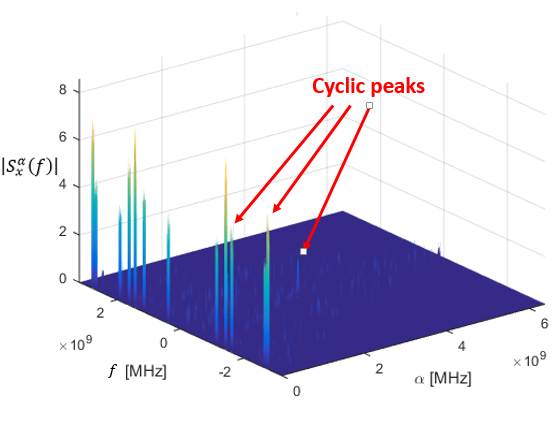}
    \caption{Reconstructed cyclic spectrum.}
    \label{fig:cyc2d1}
  \end{center}
\end{figure}
\begin{figure}[tb]
  \begin{center}
    \includegraphics[width=0.9\columnwidth]{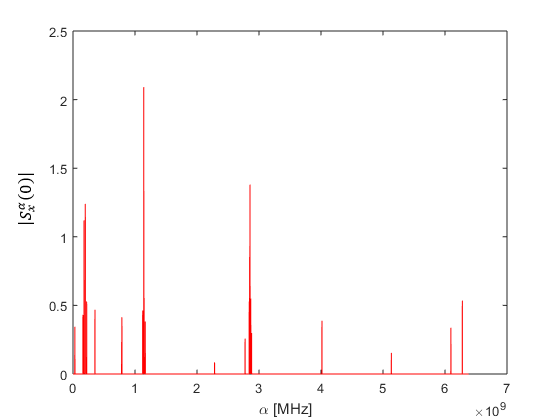}
    \caption{Reconstructed cyclic spectrum for $f=0$, $S_x^{\alpha}(0)$, as a function of the cyclic frequency $\alpha$.}
    \label{fig:cyc1d1}
  \end{center}
\end{figure}

Next, we perform cyclostationary detection on the reconstructed cyclic spectrum. 
We compare the performance of cyclostationary and energy detection performed on the reconstructed cyclic and power spectrum, respectively. For cyclostationary detection, we use a single-cycle detector which computes the energy at several frequencies around $f=0$ and at a single cyclic frequency $\alpha$. In the simulations, we consider AM modulated signals. We address a blind scenario where the carrier frequencies of the signals occupying the wideband channel are unknown and we have $N_{\text{sig}}=3$ potentially active transmissions, with single-sided bandwidth $B=100$MHz. For each iteration, the alternative and null hypotheses define the presence or absence of one out of the $N_{\text{sig}}$ transmissions. We refer to that transmission as the signal of interest.
The Nyquist rate of $x(t)$ is $f_{\text{Nyq}}=10GHz$. We consider $N=64$ spectral bands and $M=7$ analog channels, each sampling at $f_s=156$ MHz. The overall sampling rate is $Mf_s=1.09$GHz which is $182\%$ of the Landau rate and $10.9 \%$ of the Nyquist rate. Here, the theoretical minimal sampling rate is $f_{\text{min}}=960$MHz. The receiver operating characteristic (ROC) curve is shown in Fig. \ref{fig:sim3} for different SNR regimes (the averages were performed over $P=15$ time windows). Detection occurs if the presence of the signal of interest is correctly detected while false alarm is declared if a detection is claimed while the signal of interest is absent. It can be seen that cyclostationary detection outperforms energy detection in low SNR regimes, as expected. This is already known in the Nyquist regime and is now shown on samples obtained at a sub-Nyquist rate.

\begin{figure}[tb]
  \begin{center}
    \includegraphics[width=0.9\columnwidth]{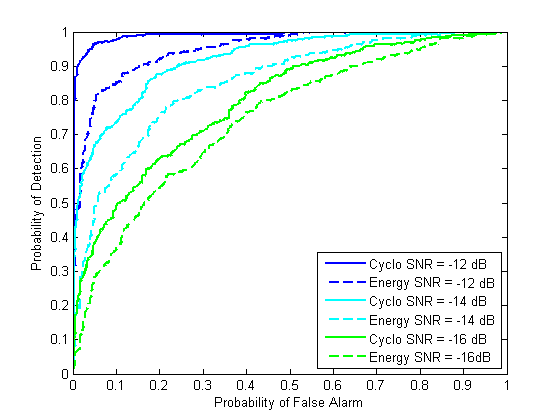}
    \caption{ROC for both energy and cyclostationary detection.}
    \label{fig:sim3}
  \end{center}
\end{figure}

\subsection{Carrier Frequencies and Bandwidths Recovery}

We now demonstrate carrier frequency and bandwidth estimation from sub-Nyquist samples. We first illustrate our algorithm process on a specific experiment. We consider $x(t)$ composed of $N_{\text{sig}}=3$ BPSK transmissions, which have cyclic features at the locations $(f, \alpha)=(0, \pm 2f_i), (\pm f_i, \pm \frac{1}{T_i})$, where $f_i$ is the carrier frequency and $T_i$ is the symbol period of the $i$th transmission \cite{GardnerBook}. 
Each transmission has bandwidth $B_i=18$MHz and symbol rate $T_i=1/B_i=0.56 \mu s$, and the carrier frequencies are drawn uniformly at random in $[0, \frac{f_{\text{Nyq}}}{2}]$, with $f_{\text{Nyq}}=1$GHz. In this experiment, the selected carriers are $f_1=163.18$MHz, $f_2=209.69$MHz and $f_3=396.12$MHz. The SNR is set to $-5$dB. In the sampling stage, we use the MWC with $M=9$ channels, each sampling at $f_s=23.26$MHz. The overall sampling rate is therefore $210$MHz, that is a little below twice the Landau rate and $21\%$ of the Nyquist rate. Here, the theoretical minimal sampling rate is $f_{\text{min}}=172.8$MHz.

Figure~\ref{fig:est_psd} presents the original and reconstructed power spectrum using $P=100$ time windows. We observe that the signal's spectrum was not perfectly recovered due to the noise. The reconstructed cyclic spectrum, including the power spectrum, estimated over $P=100$ time windows as well, is shown in Fig.~\ref{fig:est_cyc} and the section corresponding to $f=0$ can be seen in Fig.~\ref{fig:est_alpha}. The cyclic peaks at the locations $(f, \alpha)=(0, \pm 2f_i)$, for $i=1,2,3$ can be observed in both figures. In Fig.~\ref{fig:est_clust}, we illustrate the clustering stage of our algorithm as a function of $\alpha$ for $f=0$. The estimated number of clusters is $6$, yielding a correctly estimated number of signals $\hat{N}_{\text{sig}}=3$. The estimated carrier frequencies using cyclostationary based estimation are $\hat{f}_1= 162.66$MHz, $\hat{f}_2=209.19$MHz and $\hat{f}_3=395.11$MHz, and the corresponding estimated bandwidths are $\hat{B}_1=17.4$MHz, $\hat{B}_2=17.4$MHz and $\hat{B}_3=17.0$MHz. Using energy based estimation, we obtain $\hat{N}_{\text{sig}}=5$ signals, with estimated carrier frequencies $\hat{f}_1=93.04$MHz, $\hat{f}_2= 162.82$MHz, $\hat{f}_3=255.86$MHz and $\hat{f}_4=383.89$MHz, $\hat{f}_5=465.21$MHz and estimated bandwidths $\hat{B}_1=\hat{B}_2=\hat{B}_3=\hat{B}_5=23.1$MHz, $\hat{B}_4=46.3$MHz. Clearly, cyclostationary detection succeeded where energy detection failed.

\begin{figure}[tb]
  \begin{center}
    \includegraphics[width=0.9\columnwidth]{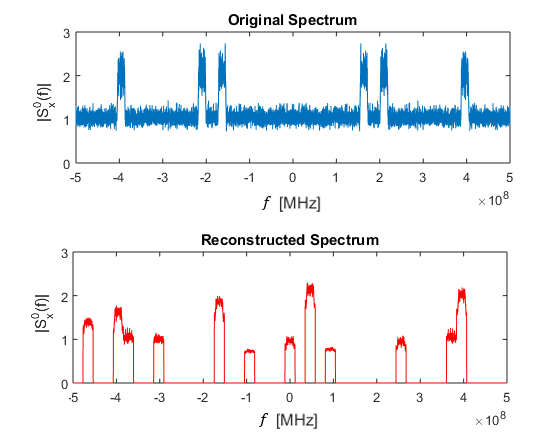}
    \caption{Original and reconstructed power spectrum.}
    \label{fig:est_psd}
  \end{center}
\end{figure}

\begin{figure}[tb]
  \begin{center}
    \includegraphics[width=0.9\columnwidth]{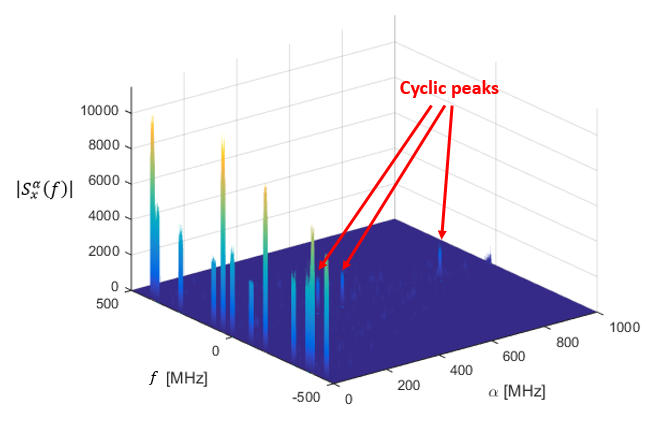}
    \caption{Reconstructed cyclic spectrum.}
    \label{fig:est_cyc}
  \end{center}
\end{figure}

\begin{figure}[tb]
  \begin{center}
    \includegraphics[width=0.9\columnwidth]{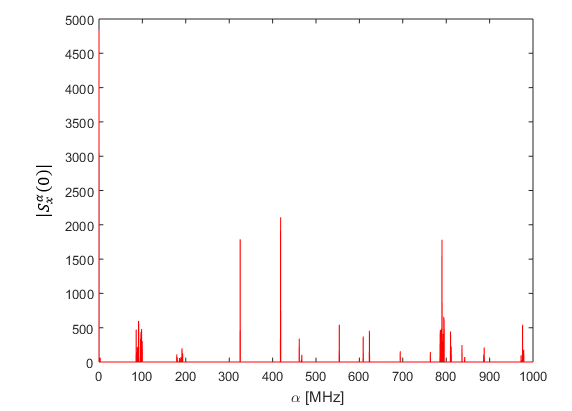}
    \caption{Reconstructed cyclic spectrum for $f=0$, $S_x^{\alpha}(0)$, as a function of the cyclic frequency $\alpha$.}
    \label{fig:est_alpha}
  \end{center}
\end{figure}

\begin{figure}[tb]
  \begin{center}
    \includegraphics[width=0.9\columnwidth]{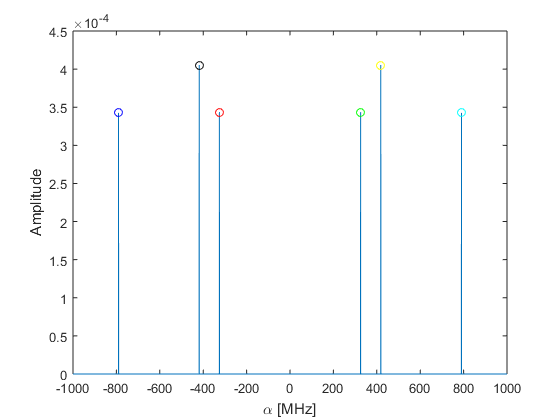}
    \caption{Clustering with $k=6$.}
    \label{fig:est_clust}
  \end{center}
\end{figure}

We now investigate the performance of our carrier frequency and bandwidth estimation algorithm from sub-Nyquist samples with respect to SNR and compare it to energy detection. We consider $x(t)$ composed of $N_{\text{sig}}=3$ BPSK transmissions with identical parameters as in the previous section. The sampling parameters remain the same as well. In each experiment, we draw the carrier frequencies uniformly at random and generate the transmissions. The results are averaged over $1000$ realizations. 

Figure.~\ref{fig:pd} shows the probability of detection of both cyclostationary (blue) and energy (red) detection. A detection is reported if the distance between the true and recovered carrier frequencies is below 10 times the frequency resolution, which is equal to $0.388$MHz. The average number of false alarms, namely unoccupied bands that are labeled as detection, is shown in Fig.~\ref{fig:pfa}. Clearly, cyclostationarity outperforms the energy approach in terms of probability of detection. Cyclostationary detection also yields fewer false alarms. For high SNRs, the gap between the performance of both schemes is small, since energy detection still succeeds in these regimes. This gap widens with SNR decrease, where the advantage of cyclostationary detection is clearly marked. The curves for both cyclostationary and energy detection show a rapid decrease of performance below a certain SNR level. We note that this level is lower for cyclostationary detection. This behavior is common to CS based recovery algorithms, which fail in the presence of large noise and yield wrong signal support, leading to misdetections and false alarms. When the SNR becomes too low, cyclostationary detection fails as well, due to the finite sensing time and averaging.

\begin{figure}[tb]
  \begin{center}
    \includegraphics[width=0.9\columnwidth]{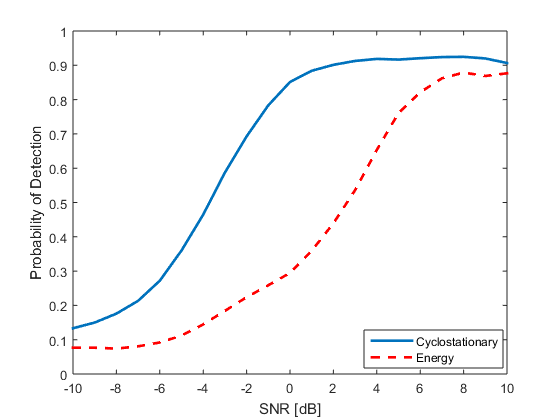}
    \caption{Probability of detection - cyclostationary vs. energy detection.}
    \label{fig:pd}
  \end{center}
\end{figure}

\begin{figure}[tb]
  \begin{center}
    \includegraphics[width=0.9\columnwidth]{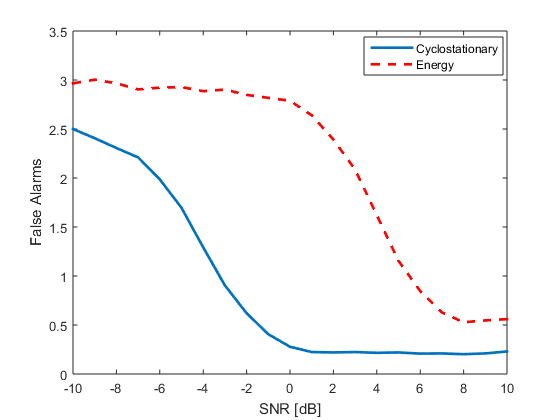}
    \caption{False alarm - cyclostationary vs. energy detection.}
    \label{fig:pfa}
  \end{center}
\end{figure}

%
Next, we compare cyclic spectrum reconstruction from Nyquist and sub-Nyquist samples. We consider the same parameters as above for the signal generation and the sampling front-end. From Fig.~\ref{fig:nyqsub}, which shows the detection performance in both regimes, it can be seen that the gap between them is not large. The loss in performance due to the reduced number of samples is small since it is compensated by cyclostationary detection, which is robust to noise. 
\begin{figure}[tb]
  \begin{center}
    \includegraphics[width=0.9\columnwidth]{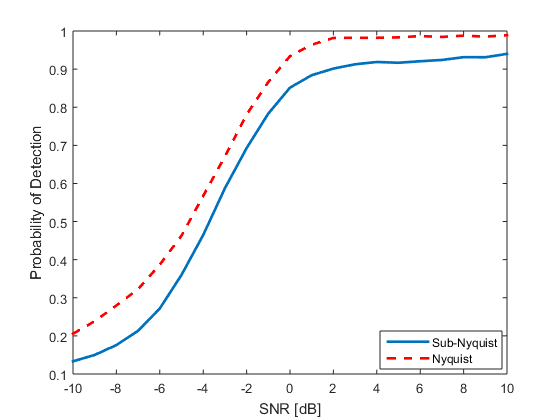}
    \caption{Probability of detection - sub-Nyquist vs. Nyquist sampling.}
    \label{fig:nyqsub}
  \end{center}
\end{figure}

In Fig.~\ref{fig:rate}, we wish to validate the derived theoretical minimal sampling rate. In the settings described above, the lower bound is $f_{\text{min}}=172.8$MHz, which corresponds to a minimal number of channels $M_{\text{min}}=10$ for perfect cyclic spectrum recovery. It can be seen in the figure that beyond $10$ channels, the probability of detection is close to $1$ in the noiseless regime. Detection errors are due to the finite sensing time and averaging. In the presence of noise, the probability of detection is slightly lower and the number of channels required to reach its maximal value is higher. Below $10$ channels, the cyclic spectrum cannot be perfectly recovered and the detection performance decreases with the number of channels.
\begin{figure}[tb]
  \begin{center}
    \includegraphics[width=0.9\columnwidth]{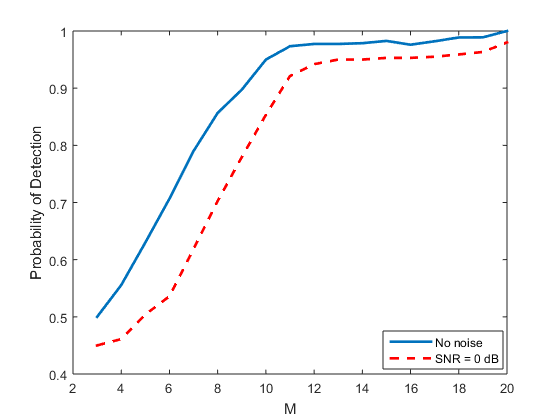}
    \caption{Probability of detection - sampling rate.}
    \label{fig:rate}
  \end{center}
\end{figure}

Finally, we compare the recovery performance of our structured OMP presented in Algorithm \ref{alg:OMP} with the traditional OMP. Here, we consider $N_{\text{sig}}=4$ transmissions and $M=14$ sampling channels. The remaining parameters are identical to those in the previous experiments. The added performance of exploiting the structure of the correlation matrices can be observed in Fig.~\ref{fig:struct} above a certain SNR value.
\begin{figure}[tb]
  \begin{center}
    \includegraphics[width=0.9\columnwidth]{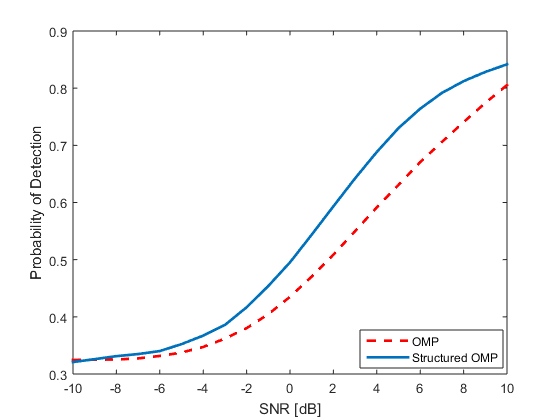}
    \caption{Probability of detection - structured OMP vs. traditional OMP.}
    \label{fig:struct}
  \end{center}
\end{figure}

\section{Conclusion}
In this paper, we considered cyclostationary detection in a sub-Nyquist regime, to cope with efficiency and robustness requirements for spectrum sensing in the context of CR. We presented a cyclic spectrum reconstruction algorithm from sub-Nyquist samples along with recovery conditions for both sparse and non sparse signals. We showed that even if the signal is not sparse, its cyclic spectrum can be recovered from samples obtained below the Nyquist rate. The minimal rates obtained for both the sparse and non sparse cases are found to be higher than those required for power spectrum recovery and lower than the rates required for signal reconstruction. Once the cyclic spectrum is recovered, we applied our feature extraction algorithm that estimates the number of transmissions and their respective carrier frequency and bandwidth. Simulations performed at low SNRs validate that cyclostationary detection outperforms energy detection in the sub-Nyquist regime, as well as the theoretical lower sampling bound.

\bibliographystyle{IEEEtran}
\bibliography{IEEEabrv,CompSens_ref}

\end{document}